\documentclass[final]{siamart1116}

\usepackage{amsmath}
\usepackage{amssymb}
\usepackage{multirow}

\newcommand{\ra}[1]{\renewcommand{\arraystretch}{#1}}

\newtheorem{example}{Example}
\newtheorem{remark}[theorem]{Remark}

\newcommand{\bb}{\mathbb}
\newcommand{\Z}{\bb{Z}}

\newcommand{\cA}{\mathcal{A}}
\newcommand{\cC}{\mathcal{C}}

\newcommand{\cH}{\mathcal{H}}
\newcommand{\cL}{\mathcal{L}}
\newcommand{\cQ}{\mathcal{Q}}

\newcommand{\cS}{\mathcal{S}}
\newcommand{\cY}{\mathcal{Y}}

\newcommand{\al}{\alpha}
\newcommand{\be}{\beta}
\newcommand{\de}{\delta}
\newcommand{\De}{\Delta}

\newcommand{\la}{\lambda}
\newcommand{\ep}{\epsilon}

\newcommand{\ups}{\upsilon}

\newcommand{\pr}{\prime}

\newcommand{\sm}{\setminus}

\newcommand{\lcm}{\text{lcm}}

\newcommand{\wt}{\text{wt}}
\newcommand{\Rad}{\text{Rad}\,}
\newcommand{\rank}{\text{rank}}

\newcommand{\ol}{\overline}

\newcommand{\lf}{\lfloor}
\newcommand{\rf}{\rfloor}

\newcommand{\Tr}{\text{Tr}}

\newcommand{\F}{\mathbb{F}}

\newcommand{\Fq}{\mathbb{F}_q}

\newcommand{\Fqm}{\mathbb{F}_{q^m}}
\newcommand{\Fqhm}{\mathbb{F}_{q^{\frac{m}{2}}}}

\newcommand{\Tqmq}{\Tr^{q^m}_{q}}
\newcommand{\Tqhmq}{\Tr^{q^{\frac{m}{2}}}_{q}}

\newcommand{\Sym}{\text{Sym}}
\newcommand{\Alt}{\text{Alt}}
\newcommand{\Qua}{\text{Qua}}
\newcommand{\PRM}{\text{PRM}}

\newcommand{\lhn}{\lfloor \frac{n}{2} \rfloor}
\newcommand{\lhm}{\lfloor \frac{m}{2} \rfloor}
\newcommand{\ltm}{\lfloor \frac{m}{3} \rfloor}
\newcommand{\hm}{\frac{m}{2}}
\newcommand{\hmpo}{\frac{m+1}{2}}
\newcommand{\hmpt}{\frac{m+2}{2}}
\newcommand{\hmmo}{\frac{m-1}{2}}
\newcommand{\hmmt}{\frac{m-2}{2}}

\newcommand{\code}{\cC_{(q,m,\de)}}
\newcommand{\gene}{g_{(q,m,\de)}(x)}
\newcommand{\codespe}{\cC_{(q,m,\de_i)}}

\numberwithin{theorem}{section}

\title{The Minimum Distance of Some Narrow-Sense Primitive BCH Codes}

\author{Shuxing Li%
  \thanks{Department of Mathematics, Simon Fraser University, 8888 University Drive, Burnaby BC V5A 1S6, Canada (\email{lsxlsxlsx1987@gmail.com}).}%
  }

\headers{The Minimum Distance of Narrow-Sense Primitive BCH Codes}{Shuxing Li}

\ifpdf
\hypersetup{
  pdftitle={The Minimum Distance of Narrow-Sense Primitive BCH Codes},
  pdfauthor={Shuxing Li}
}
\fi

\begin{document}
\maketitle

\begin{abstract}
Due to wide applications of BCH codes, the determination of their minimum distance is of great interest. However, this is a very challenging problem for which few theoretical results have been reported in the last four decades. Even for the narrow-sense primitive BCH codes, which form the most well-studied subclass of BCH codes, there are very few theoretical results on the minimum distance. In this paper, we present new results on the minimum distance of narrow-sense primitive BCH codes with special Bose distance. We prove that for a prime power $q$, the $q$-ary narrow-sense primitive BCH code with length $q^m-1$ and Bose distance $q^m-q^{m-1}-q^i-1$, where $\hmmt \le i \le m-\ltm-1$, has minimum distance $q^m-q^{m-1}-q^i-1$. This is achieved by employing the beautiful theory of sets of quadratic forms, symmetric bilinear forms and alternating bilinear forms over finite fields, which can be best described using the framework of association schemes.
\end{abstract}



\section{Introduction}

BCH codes are the most important class of cyclic codes. Due to their efficient encoding and decoding algorithms, BCH codes are widely used in error correction, communication and data storage. Despite the various applications of BCH codes, there remain many challenging problems in the theoretical aspects of BCH codes.

A fundamental theoretical problem of BCH codes is the determination of the basic parameters, including the dimension and the minimum distance. Although there have been a series of literature in this area \cite{AKS,ACS,AS,Ber67,Ber70,Ber15,CC,Char90,Char98,Co,Ding2,Ding1,Ding3,K1,K2,KL,KLP1,KT,MS,Mann,Pe,YF}, this problem is still wide open in general. In particular, the determination of minimum distance is very difficult and there are few known theoretical results \cite{Ber70,Ding2,Ding1,K2,KL,KLP1,Pe}. For numerical results on the minimum distance, please refer to \cite[Section 3.3]{Char98}.

Let $\Fqm$ be a finite field. Let $n$ be a divisor of $q^m-1$ and $\be$ be an element of $\Fqm$ with multiplicative order $n$. Roughly speaking, a BCH code is a cyclic code with length $n$, whose generator polynomial has a set of $\de-1$ consecutive roots $\be^b,\be^{b+1},\ldots,\be^{b+\de-2}$, where $b$ is a positive integer. By the BCH bound, the minimum distance of the BCH code is at least $\de$. Thus, we say this BCH code has {\it designed distance} $\de$. A BCH code is called {\it primitive}, if the length $n=q^m-1$. A BCH code is called {\it narrow-sense}, if $b=1$, i.e., the $\de-1$ consecutive roots start from $\be$. For a narrow-sense BCH code, its largest possible designed distance is called the {\it Bose distance} \cite[p. 281]{Ber15}, \cite[p. 205]{MS} and denoted by $d_B$. Narrow-sense primitive BCH codes form the most well-studied subclass of BCH codes. An attractive property of narrow-sense primitive BCH codes is that the Bose distance provides a generally good lower bound on the minimum distance. For instance, it is conjectured that for binary narrow-sense primitive BCH codes, the minimum distance $d$ satisfies $d_B \le d \le d_B+4$ \cite[p. 1011]{Char98}. In general, it is extremely hard to determine the minimum distance of BCH codes with arbitrary Bose distance. On the other hand, when the Bose distance has a certain special form, there are some known results on minimum distance.


\begin{theorem}\label{thm-known}
Let $q$ be a prime power. Given a $q$-ary narrow-sense primitive BCH code with length $q^m-1$ and Bose distance $d_B$, its minimum distance is known in the following cases:
\begin{itemize}
\item[1)] if $d_B=q^i-1$ with $1 \le i \le m-1$, then $d=q^i-1$ {\rm \cite[Theorem 5]{Pe}}
\item[2)] if $q=2$ and $d_B=2^{m-1}-2^i-1$ with $\hmmt \le i \le m-\ltm-1$, then $d=2^{m-1}-2^i-1$ {\rm \cite[Corollary]{Ber70}}
\item[3)] if $q=2$ and $d_B=2^{m-1-s}-2^{m-1-i-s}-1$ with $1 \le i \le m-s-2$ and $0 \le s \le m-2i$, then $d=2^{m-1-s}-2^{m-1-i-s}-1$ {\rm \cite[Corollary 1]{KL}}
\end{itemize}
\end{theorem}

We remark that the above results can be dated back to four decades ago. To our best knowledge, very few results concerning the minimum distance of BCH codes are obtained since then. Very recently, Ding, Du, and Zhou studied $q$-ary narrow-sense primitive BCH codes of length $q^m-1$ with designed distance $q^i+1$ \cite{Ding1}. They derived the Bose distances of these BCH codes and showed that the minimum distances coincide with the Bose distances, when $(m-i) \mid i$ or $2i \mid m$ \cite[Theorem 13]{Ding1}. Ding studied $q$-ary narrow-sense primitive BCH codes of length $q^m-1$, whose Bose distances have the form $(q-\ell_0)q^{m-\ell_1-1}-1$, where $0 \le \ell_0 \le q-2$ and $0 \le \ell_1 \le m-1$ \cite{Ding2}. He derived that the minimum distances of these BCH codes coincide with the Bose distances \cite[Theorem 10]{Ding2}.

We note that 2) of Theorem~\ref{thm-known} follows from the classical results on the weight distribution of subcodes of the second order Reed-Muller codes due to Berlekamp \cite[Corollary]{Ber70} (see also Kasami \cite{K2}). A major motivation of this paper is to generalize 2) of Theorem~\ref{thm-known} to $q$-ary narrow-sense primitive BCH codes for an arbitrary prime power $q$. As a consequence, we have the following main theorem.

\begin{theorem}\label{thm-main}
Let $q$ be a prime power and $m$ be a positive integer satisfying
$$
\begin{cases}
  m \ge 3 & \mbox{if $q=2$}, \\
  m \ge 2 & \mbox{if $q=3$}, \\
  m \ge 1 & \mbox{if $q \ge 4$}.
\end{cases}
$$
For nonnegative integer $i$ with $\frac{m-2}{2} \le i \le m-\ltm-1$, the $q$-ary narrow-sense primitive BCH code with length $q^m-1$ and Bose distance $q^m-q^{m-1}-q^i-1$ has minimum distance $q^m-q^{m-1}-q^i-1$.
\end{theorem}

We remark that when $q \in \{2,3\}$, the restrictions on $m$ ensure that the code under consideration is well-defined.  In the following, we provide an outline of our proof. The concepts and terminologies used in the outline, will be formally defined later.
\vspace{5pt}
\begin{itemize}
\item[1)] In Section 2.2, we show that the lower bound on $i$ guarantees that $q^m-q^{m-1}-q^i-1$ is indeed the Bose distance. In order to derive the minimum distance, it suffices to show that the narrow-sense primitive BCH code $\cC$ contains a codeword having weight $q^m-q^{m-1}-q^i-1$.
\item[2)] In Section 2.3, using the trace representation of the narrow-sense primitive BCH code $\cC$, we show that the upper bound on $i$ guarantees that $\cC$ is a subcode of the second order punctured Reed-Muller code (when $q=2$) or a subcode of the second order punctured generalized Reed-Muller code (when $q>2$). Moreover, in Proposition~2.6, we observe that the code $\cC$ can be decomposed into a disjoint union of cosets of the first order punctured (generalized) Reed-Muller code, where each coset representative corresponds to a quadratic form over finite field $\Fqm$. More specifically, we have
    $$
    \cC=\bigcup_{Q \in \cQ}\big((Q(x))_{x \in \Fqm^*}+\PRM_q(1,m)\big),
    $$
    where $\cQ$ is a set of quadratic forms on $\Fqm$ and $\PRM_q(1,m)$ is the first order punctured (generalized) Reed-Muller code.
\item[3)] In Propositions~\ref{prop-Qweightodd} and \ref{prop-Qweighteven}, we prove that the weight distribution of the code $(Q(x))_{x \in \Fqm^*}+\PRM_q(1,m)$, where $Q$ is a quadratic form on $\Fqm$, depends only on the rank and type of $Q$. Thus, we aim to find a quadratic form $Q \in \cQ$ with proper rank and type, such that $(Q(x))_{x \in \Fqm^*}+\PRM_q(1,m)$ contains a codeword having weight $q^m-q^{m-1}-q^i-1$. If so, the minimum distance of $\cC$ is equal to $q^m-q^{m-1}-q^i-1$.
\item[4)] In general, determining the rank and type of an individual quadratic form $Q \in \cQ$ is a challenging problem. Alternatively, we regard the set of quadratic forms $\cQ$ as a whole and study the quadratic forms belonging to $\cQ$ collectively, in the context of association schemes. We introduce some background knowledge of association schemes and the association schemes formed by quadratic forms in Sections 3.1 and 3.4, respectively. If we regard $\cQ$ as a subset of the association scheme formed by quadratic forms, our task amounts to derive some information about the inner distribution of the subset $\cQ$.
\item[5)] Recall that each quadratic form $Q$ produces an associated symmetric bilinear form, when $q$ is odd, and an associated alternating bilinear form, when $q$ is even. Therefore, starting from the set $\cQ$, we can derive a set of symmetric bilinear forms $\cS$, when $q$ is odd, and a set of alternating bilinear forms $\cA$, when $q$ is even. $\cS$ and $\cA$ can be viewed as a subset of the association scheme formed by symmetric bilinear forms (Section 3.2) and a subset of the association scheme formed by alternating bilinear forms (Section 3.3), respectively. Proposition 3.11 suggests that the inner distribution of $\cQ$ has a close connection with that of $\cS$ and $\cA$. Thus, we reduce the original problem to the investigation of inner distribution of $\cS$ and $\cA$.
\item[6)] If the subset $\cS$ and $\cA$ satisfy certain nice properties, i.e., the subset is a code and design in the association scheme, the inner distribution of $\cS$ and $\cA$, has been determined by Schmidt \cite{Sch15} (see Propositions 3.1 and 3.2), and Delsarte and Goethals \cite{DG} (see Proposition 3.3), respectively. What remains is to verify that $\cS$ and $\cA$ indeed satisfy the nice properties, which is completed in Propositions 4.6 and 5.2, respectively. Using Proposition 3.11, we derive the information concerning the inner distribution of $\cQ$, from that of $\cS$ (see Theorem 4.7) and $\cA$ (see Theorem 5.3), and accomplish the proof. In addition, we can obtain the complete information of the inner distribution of $\cQ$, when $q$ is odd. Thus, we derive the weight distribution of $\cC$, when $q$ is odd (Theorem 4.7).
\end{itemize}

The rest of this paper is organized as follows. In Section~\ref{sec2}, we present some preliminaries on narrow-sense primitive BCH codes. Section~\ref{sec3} provides a brief introduction to association scheme, including some crucial facts about association schemes of quadratic forms, symmetric bilinear forms and alternating bilinear forms over finite fields. The proof of our main theorem is presented in Section~\ref{sec4} and Section~\ref{sec5}, for $q$ being odd and even respectively. Section~\ref{sec6} is devoted to some concluding remarks.

\section{Preliminaries on narrow-sense primitive BCH codes}\label{sec2}

\subsection{Narrow-sense primitive BCH codes}

Let $q$ be a prime power. Let $\cC$ be an $[n,k]$ linear code over finite field $\Fq$ with $\gcd(n,q)=1$. $\cC$ is called a cyclic code, if $(c_0,c_1,\ldots,c_{n-1}) \in \cC$ implies its cyclic shift $(c_{n-1},c_0,\ldots,c_{n-2}) \in \cC$. For a cyclic code $\cC$, each codeword $(c_0,\ldots,c_{n-1})$ can be associated with a polynomial $\sum_{i=0}^{n-1} c_ix^i$ in the principal ideal ring $R_n=\Fq[x]/(x^{n}-1)$. Under this correspondence, $\cC$ can be identified with an ideal of $R_n$. Hence, there is a unique monic polynomial $g(x) \in \Fq[x]$ with $g(x) \mid x^n-1$ such that $\cC=(g(x))R_n$ and $g(x)$ has the smallest degree among the elements in $\cC$. This $g(x)$ is called the {\em generator polynomial} of $\cC$, and $h(x)=\frac{x^n-1}{g(x)}$ is called the {\em parity-check polynomial} of $\cC$. When $R_n$ is specified, a cyclic code is uniquely determined by either the generator polynomial or the parity-check polynomial and the dimension is equal to $n-\deg(g(x))=\deg(h(x))$. $\cC$ is said to have $i$ {\em nonzeroes} if its parity-check polynomial can be factorized into a product of $i$ irreducible polynomials over $\Fq$.

Now we introduce some notation which will be used throughout the rest of the paper. Let $q$ be a prime power. Set $\al$ to be a primitive element of finite field $\Fqm$. For $0 \le i \le q^m-2$, define $m_i(x)$ to be the minimum polynomial of $\al^i$ over $\Fq$. For an integer $2 \le \de \le q^m-1$, define a polynomial $\gene$ over $\Fq$ as
$$
\gene=\lcm(m_1(x),m_2(x),\ldots,m_{\de-1}(x)),
$$
where $\lcm$ represents the least common multiple of $m_i(x)$, $1 \le i \le \de-1$. Let $R$ be the principal ideal ring $\Fq[x]/(x^{q^m-1}-1)$. Then the cyclic code $(\gene)R$ is called a {\it $q$-ary narrow-sense primitive BCH code} with length $q^m-1$ and {\it designed distance} $\de$, which is denoted by $\cC_{(q,m,\de)}$. The dimension of $\cC_{(q,m,\de)}$ is equal to $q^m-1-\deg(\gene)$. By the BCH bound, the minimum distance of $\cC_{(q,m,\de)}$ is no less than its designed distance $\de$. For two distinct integers $2 \le \de,\de^{\pr} \le q^m-1$, two narrow-sense primitive BCH codes with designed distances $\de$ and $\de^{\pr}$ may coincide. Thus, for a narrow-sense primitive BCH code, its {\it Bose distance} \cite[p. 281]{Ber15}, \cite[p. 205]{MS} is defined to be the largest designed distance and denoted by $d_B$.

For a code $\cC$ with length $n$, we use $A_i$ to denote the number of codewords in $\cC$ with Hamming weight $i$, where $0 \le i \le n$. Then the sequence $(A_0,A_1,A_2,\ldots,A_n)$ is called the {\it weight distribution} of $\cC$. Moreover, we can use a polynomial
$$
\sum_{i=0}^n A_iZ^i
$$
to represent the weight distribution of $\cC$ in a compact way, and this polynomial is called the {\it weight enumerator} of $\cC$.

\subsection{Cyclotomic cosets}

Let $q$ be a prime power. We use $\Z_{q^m-1}$ to denote the ring of integers modulo $q^m-1$. Let $s$ be an integer with $0 \leq s < q^m-1$. The {\it $q$-cyclotomic coset of $s$ modulo $q^m-1$} is defined by
$$
C_s=\{sq^j \bmod{(q^m-1)} \mid 0 \le j \le l_s-1\} \subset \Z_{q^m-1},
$$
where $l_s$ is the smallest positive integer such that $q^{l_s} s \equiv s  \pmod{q^m-1}$. We call the smallest integer in $C_s$ as the {\it coset leader} of $C_s$. We use $\Gamma_{(q^m-1,q)}$ to denote the set of all the coset leaders of $q$-cyclotomic cosets modulo $q^m-1$.

Recall that $\al$ is a primitive element of $\Fqm$ and $m_s(x)$ is the minimal polynomial of $\al^s$ over $\Fq$, $0 \le s \le q^m-2$. Then, we have
$$
m_s(x)=\prod_{i \in C_s}(x-\al^i).
$$
Thus, we know that
\begin{equation}\label{eqn-deg}
\deg(m_s(x))=|C_s|.
\end{equation}
The following proposition shows the dimension of a BCH code is a summation of the size of certain cyclotomic cosets.

\begin{proposition}\label{prop-dim}
The dimension of BCH code $\code$ is equal to
$$
\sum_{\substack{s \ge \de \\ s \in \Gamma_{(q^m-1,q)}}} |C_s|+1.
$$
\end{proposition}
\begin{proof}
Note that the dimension of $\code$ is equal to the degree of its parity-check polynomial $h(x)$. By the definition of BCH code,
$$
h(x)=(x-1)\prod_{\substack{s \ge \de \\ s \in \Gamma_{(q^m-1,q)}}} m_s(x).
$$
The conclusion follows from (\ref{eqn-deg}).
\end{proof}

For any integer $0 \le s < q^m-1$, it has a unique $q$-adic expansion $s=\sum_{i=0}^{m-1}s_iq^i$, where $0 \le s_i \le q-1$ for $0 \le i \le m-1$. Define the {\it $q$-weight} of $s$ to be $w_q(s)=\sum_{i=0}^{m-1} s_i$. Define the sequence of $s$ to be
$$
\ol{s}=(s_{m-1},s_{m-2},s_{m-3},\ldots,s_0).
$$
For $1 \le j \le m-1$, it is easy to see that the sequence of $q^js$ is a cyclic shift of $\ol{s}$:
$$
\ol{q^js}=(s_{m-j-1},s_{m-j-2},\ldots,s_{m-j+1},s_{m-j}),
$$
where the subscript of each coordinate is regarded as an integer modulo $m$. Given two sequences $\ol{s}=(s_{m-1},s_{m-2},\ldots,s_0)$ and $\ol{t}=(t_{m-1},t_{m-2},\ldots,t_0)$, define $\ol{s}<\ol{t}$ (resp. $\ol{s}\le\ol{t}$) if $\sum_{i=0}^{m-1}s_iq^i < \sum_{i=0}^{m-1}t_iq^i$ (resp. $\sum_{i=0}^{m-1}s_iq^i \le \sum_{i=0}^{m-1}t_iq^i$). Thus, $s$ is a coset leader if and only if
$$
\ol{s} \le \ol{q^js}, \quad\forall 1 \le j \le m-1.
$$

The following lemma demonstrates the first few largest coset leaders belonging to $\Gamma_{(q^m-1,q)}$. Throughout the rest of this paper, we denote $\de=(q-1)q^{m-1}-1$ and $\de_i=(q-1)q^{m-1}-q^i-1$ for $0 \le i \le m-1$.

\begin{lemma}\label{lem-cosetleader}
For $\frac{m-2}{2} \le i \le m-\lf \frac{m}{3} \rf-1$, the set $\{\de, \de_j \mid \frac{m-2}{2} \le j \le i \}$ consists of all coset leaders in $\Gamma_{(q^m-1,q)}$, which are greater than or equal to $\de_i$.
\end{lemma}
\begin{proof}
Let $i$ and $j$ be two integers $\frac{m-2}{2} \le j \le i \le m-\lf \frac{m}{3} \rf-1$. Define $A$ to be the subset of $\Gamma_{(q^m-1,q)}$ consisting of all coset leaders greater than or equal to $\de_i$.

Suppose $0 \le \theta < q^m-1$ is a coset leader, which is greater than or equal to $\de_i$. Then $0 \le w_q(\theta) \le m(q-1)-1$. To find the few largest coset leaders, we focus on the coset leaders with large $q$-weight.

When $w_q(\theta)=m(q-1)-1$, it is easy to see that $\ol{\theta}$ must be of the following form:
$$
\ol{\theta}=(q-2,q-1,q-1,\ldots,q-1).
$$
Thus, we have $\theta=\de$ and $\de \in A$.

When $w_q(\theta)=m(q-1)-2$, then $\ol{\theta}$ must be one of the following forms:
\begin{equation}\label{eqn-leader1}
\ol{\theta}=(q-3,q-1,q-1,\ldots,q-1),
\end{equation}
or
\begin{equation}\label{eqn-leader2}
\ol{\theta}=(\underset{m-1}{q-2},\underbrace{q-1,\ldots,q-1}_{m-2-j},\underset{j}{q-2},\underbrace{q-1,\ldots,q-1}_{j}).
\end{equation}
If $\ol{\theta}$ has form (\ref{eqn-leader1}), then $\theta < \de_i$. If $\ol{\theta}$ has form (\ref{eqn-leader2}), we must have $m-2-j \le j$, namely, $j \ge \frac{m-2}{2}$. Otherwise, if $m-2-j > j$, then $\ol{q^{m-1-j}\theta}<\ol{\theta}$, which leads to a contradiction. Meanwhile, to ensure $\theta \ge \de_i$, we must have $j \le i$. Hence, we have $\{\de_j \mid \frac{m-2}{2} \le j \le i\} \subset A$.

When $w_q(\theta)=m(q-1)-3$, then $\ol{\theta}$ must be one of the following forms:
\begin{equation}\label{eqn-leader3}
\ol{\theta}=(q-4,q-1,q-1,\ldots,q-1),
\end{equation}
or
\begin{equation}\label{eqn-leader4}
\ol{\theta}=(q-3,q-1,\ldots,q-1,q-2,q-1,\ldots,q-1),
\end{equation}
or
\begin{equation}\label{eqn-leader5}
\ol{\theta}=(\underset{m-1}{q-2},\underbrace{q-1,\ldots,q-1}_{m-2-j_1},\underset{j_1}{q-2},\underbrace{q-1,\ldots,q-1}_{j_1-j_2-1},\underset{j_2}{q-2},\underbrace{q-1,\ldots,q-1}_{j_2}).
\end{equation}
If $\ol{\theta}$ has form (\ref{eqn-leader3}) or (\ref{eqn-leader4}), then $\theta < \de_i$. If $\ol{\theta}$ has form (\ref{eqn-leader5}), we must have $m-2-j_1 \le j_1-j_2-1$ and $m-2-j_1 \le j_2$. Since $(m-2-j_1)+(j_1-j_2-1)+j_2=m-3$, we have $m-2-j_1 \le \lf \frac{m-3}{3}\rf$, namely, $j_1 \ge m-\lf \frac{m}{3} \rf-1$. Note that
$$
\ol{\de_i}=(\underset{m-1}{q-2},\underbrace{q-1,\ldots,q-1}_{m-2-i},\underset{i}{q-2},q-1,\ldots,q-1).
$$
Since $i \le m-\lf \frac{m}{3} \rf-1$, we have $\theta < \de_i$. Therefore, each coset leader with $q$-weight $m(q-1)-3$ is less than $\de_i$.

When $w_q(\theta)<m(q-1)-3$, a similar approach shows that $\theta<\de_i$. Thus, we have proved that $A=\{\de,\de_j \mid \frac{m-2}{2} \le j \le i\}$.
\end{proof}

\begin{remark}\label{rem-cosetleader}
Suppose $m-\lf \frac{m}{3} \rf-1 < i \le m-1$, then there exists a coset leader $\theta$ satisfying $w_q(\theta)=(m-1)q-3$ and
$$
\ol{\theta}=(\underset{m-1}{q-2},\underbrace{q-1,\ldots,q-1}_{\lf\frac{m}{3}\rf-1},\underset{m-\lf\frac{m}{3}\rf-1}{q-2},\underbrace{q-1,\ldots,q-1}_{\lf \frac{m}{3}\rf-1},\underset{m-2\lf\frac{m}{3}\rf-1}{q-2},\underbrace{q-1,\ldots,q-1}_{m-2\lf\frac{m}{3}\rf-1}),
$$
such that $\theta>\de_i$. Hence, the condition $i \le m-\lf \frac{m}{3} \rf-1$ in Lemma~\ref{lem-cosetleader} guarantees that each coset leader greater than or equal to $\de_i$ has $q$-weight either $(m-1)q-1$ or $(m-1)q-2$.
\end{remark}

\subsection{Trace representation of cyclic codes}

Denote the trace function from $\Fqm$ to $\Fq$ by $\Tqmq$. Now we recall the following trace representations of cyclic codes, which is a direct consequence of Delsarte's Theorem \cite{Del75}.

\begin{proposition}\label{prop-trace}
Let $\cC$ be a cyclic code of length $q^m-1$ over $\Fq$. Suppose $\cC$ has $s$ nonzeroes and let $\al^{i_1}, \al^{i_2}, \ldots, \al^{i_s}$ be $s$ roots of its parity-check polynomial which are not conjugate with each other. Let $C_{i_j}$ be the $q$-cyclotomic coset modulo $q^m-1$ containing $i_j$. Denote the size of $C_{i_j}$ to be $m_j$, $1 \le j \le s$. Then $\cC$ has the following trace representation
$$
\cC=\{ c(\la_1,\la_2,\ldots,\la_s) \mid \la_j \in \F_{q^{m_j}}, 1 \le j \le s \},
$$
where
$$
c(\la_1,\la_2,\ldots,\la_s)=\left(\sum_{j=1}^s \Tr^{q^{m_j}}_{q}(\la_j\al^{-li_j})\right)_{l=0}^{q^m-2}.
$$
\end{proposition}

Consequently, we have the following proposition concerning BCH code $\cC_{(q,m,\de_i)}$, where $\frac{m-2}{2} \le i \le m-\ltm-1$.

\begin{proposition}\label{prop-traceprop}
For nonnegative integer $i$ with $\frac{m-2}{2} \le i \le m-\ltm-1$, the BCH code $\cC_{(q,m,\de_i)}$ has dimension $(i-\frac{m-5}{2})m+1$ and Bose distance $\de_i$. When $m$ is odd, it has the following trace representation:
$$
\Bigg\{\Big(\Tqmq\Big(\sum_{j=\hmpo}^{i+1}\la_jx^{q^j+1}+\mu x\Big)+\ep\Big)_{x\in \Fqm^*} \Big | \la_{\hmpo},\ldots,\la_{i+1}, \mu \in \Fqm, \ep \in \Fq \Bigg\}.
$$
When $m$ is even, it has the following trace representation:
\begin{align*}
\Bigg\{\Big(\Tqhmq\big(\la_{\hm}x^{q^{\hm}+1}\big)+\Tqmq&\Big(\sum_{j=\hmpt}^{i+1}\la_jx^{q^j+1}+\mu x\Big)+\ep\Big)_{x\in \Fqm^*} \\
                  &\Big | \la_{\hm} \in \Fqhm, \la_{\hmpt},\ldots,\la_{i+1}, \mu \in \Fqm, \ep \in \Fq \Bigg\}.
\end{align*}
\end{proposition}
\begin{proof}
We only prove the case with $m$ being even. The case with $m$ being odd can be shown in a similar way. By Lemma~\ref{lem-cosetleader}, $\{\de,\de_j \mid \hmmt \le j \le i \}$ is the set of all coset leaders in $\Gamma_{(q^m-1,q)}$, which are greater or equal to $\de_i$. By analyzing the sequences of coset leaders $\de$ and $\de_j$, $\hmmt \le j \le i$, it is easy to see that $|C_{\de_{\hmmt}}|=\hm$ and $|C_{\de}|=|C_{\de_j}|=m$ for $\hm \le j \le i$. Therefore, by Proposition~\ref{prop-dim}, the dimension of $\codespe$ is equal to $(i-\frac{m-5}{2})m+1$.

Since $\codespe$ is a narrow-sense BCH code and $\de_i$ is a coset leader, by \cite[Proposition 4]{LDXG}, its Bose distance is equal to $\de_i$. By definition, $\codespe$ has $i-\hm+4$ nonzeroes. More precisely, $\{ 1,\al^{\de},\al^{\de_{j}} \mid \hmmt \le j \le i\}$ are $i-\hm+4$ roots of the parity-check polynomial of $\codespe$, which are not conjugate with each other. Recall that the trace function satisfies $\Tqmq(x^q)=\Tqmq(x)$ and $\Tqhmq(x^q)=\Tqhmq(x)$. Since
\begin{align*}
q(q^m-1-\de) &\equiv 1 \pmod{q^m-1}, \\
q(q^m-1-\de_j) & \equiv q^{j+1}+1 \pmod{q^m-1},
\end{align*}
employing Proposition~\ref{prop-trace}, we get the trace representation of $\codespe$.
\end{proof}

When $q=2$ (resp. $q>2$), the punctured first order Reed-Muller code (resp. the $q$-ary punctured first order generalized Reed-Muller code) {\rm\cite[Chapter 5]{AK}} with length $q^m-1$ is denoted by $\PRM_q(1,m)$. It has the following trace representation:
\begin{equation}\label{eqn-firstord}
\PRM_q(1,m)=\Big\{\big(\Tqmq(\mu x)+\ep\big)_{x\in \Fqm^*} \big | \mu \in \Fqm, \ep \in \Fq \Big\}.
\end{equation}
It is easy to verify that $\PRM_q(1,m)$ has weight enumerator
\begin{equation}\label{eqn-weightenum}
T(Z)=1+(q-1)(q^m-1)Z^{q^{m}-q^{m-1}-1}+(q^m-1)Z^{q^m-q^{m-1}}+(q-1)Z^{q^m-1}.
\end{equation}
According to Proposition~\ref{prop-traceprop}, for $\hmmt \le i \le m-\ltm-1$, the BCH code $\codespe$ is a disjoint union of cosets of $\PRM_q(1,m)$. More precisely, for $m$ being odd, define
\begin{equation}\label{eqn-Q1}
\cQ_1:=\cQ_1(i)=\Bigg\{\Tqmq\Big(\sum_{j=\hmpo}^{i+1}\la_jx^{q^j+1}\Big) \Big | \la_{\hmpo},\ldots,\la_{i+1} \in \Fqm \Bigg\}.
\end{equation}
For $m$ being even, define
\begin{equation}
\begin{split}\label{eqn-Q2}
\cQ_2:=\cQ_2(i)=\Bigg\{\Tqhmq\big(\la_{\hm}x^{q^{\hm}+1}\big)+\Tqmq&\Big(\sum_{j=\hmpt}^{i+1}\la_jx^{q^j+1}\Big) \\
                                                             &\Big | \la_{\hm} \in \Fqhm, \la_{\hmpt},\ldots,\la_{i+1} \in \Fqm \Bigg\}.
\end{split}
\end{equation}
The we have the following description of $\codespe$.

\begin{proposition}\label{prop-cosetunion}
For nonnegative integer $i$ with $\frac{m-2}{2} \le i \le m-\ltm-1$, we have
\begin{equation}\label{eqn-codeQ}
\codespe=\begin{cases}
            \bigcup_{Q \in \cQ_1}\big((Q(x))_{x\in\Fqm^*}+\PRM_q(1,m)\big) & \mbox{if $m$ is odd,} \\
            \bigcup_{Q \in \cQ_2}\big((Q(x))_{x\in\Fqm^*}+\PRM_q(1,m)\big) & \mbox{if $m$ is even.}
         \end{cases}
\end{equation}
Moreover, when $q=2$ (resp. $q>2$), the BCH code $\codespe$ is a subcode of the punctured second order Reed-Muller code (resp. the $q$-ary punctured second order generalized Reed-Muller code).
\end{proposition}
\begin{proof}
Combining Proposition~\ref{prop-traceprop} and (\ref{eqn-firstord}), (\ref{eqn-Q1}), (\ref{eqn-Q2}), we can derive (\ref{eqn-codeQ}).
By employing the same argument as in {\rm\cite[p. 5334]{LF}}, we can see the elements of $\cQ_1$ and $\cQ_2$ are all quadratic forms over $\Fq^m$, which is the additive group of the finite field $\Fqm$ (see Section~\ref{sec3D} for a formal definition of quadratic forms over finite fields). Hence, when $q=2$, $\codespe$ is a subcode of the punctured second order Reed-Muller code. When $q>2$, $\codespe$ is a subcode of the $q$-ary punctured second order generalized Reed-Muller code.
\end{proof}

\begin{remark}
By Remark~\ref{rem-cosetleader}, the upper bound $i\le m-\ltm-1$ ensures that in $\Gamma_{(q^m-1,q)}$, all coset leaders greater than or equal to $\de_i$ has $q$-weight either $(q-1)m-1$ or $(q-1)m-2$. This essentially guarantees that $\codespe$ is a subcode of the punctured second order Reed-Muller code or the punctured second order generalized Reed-Muller code.
\end{remark}

By Proposition~\ref{prop-cosetunion}, $\codespe$ is a disjoint union of cosets of $\PRM_{q}(1,m)$ and the coset representatives are of the form $(Q(x))_{x\in\Fqm^*}$, where $Q \in \cQ_1$ or $Q \in \cQ_2$. In order to determine the minimum distance of $\codespe$, it suffices to study the minimum distance of the subcode
$$
(Q(x))_{x \in \Fqm^*}+\PRM_q(1,m),
$$
for each quadratic form $Q \in \cQ_1$ or $Q \in \cQ_2$. As we will see in Propositions~\ref{prop-Qweightodd} and \ref{prop-Qweighteven}, the minimum distance of $(Q(x))_{x \in \Fqm^*}+\PRM_q(1,m)$ depends only on the rank and type of $Q$. Suppose there exists a quadratic form $Q \in \cQ_1$ or $Q \in \cQ_2$ having proper rank and type, such that $(Q(x))_{x \in \Fqm^*}+\PRM_q(1,m)$ contains a codeword with weight $\de_i$. Then the minimum distance is obtained. We note that determining the rank and type of an individual quadratic form is a challenging problem in general. Alternatively, we regard the set of quadratic forms $\cQ_1$ or $\cQ_2$ as a whole and study them in the context of association schemes. The property of $\cQ_1$ or $\cQ_2$, ensures that the desired quadratic form with proper rank and type, does belong to $\cQ_1$ or $\cQ_2$. The next section is devoted to background knowledge on symmetric bilinear forms, alternating bilinear forms and quadratic forms over finite fields, as well as their related association schemes.

\section{Association schemes of symmetric bilinear forms, alternating bilinear forms and quadratic forms over finite fields}\label{sec3}

\subsection{Association schemes}

Association scheme plays a central role in the study of algebraic combinatorics \cite{BI,Del73}. In this subsection, we briefly review some concepts related to association schemes.

Let $X$ be a finite set and $R_0, R_1,\ldots,R_n$ be a partition of $X \times X$. Suppose
\begin{itemize}
\item[1)] $R_0=\{(x,x) \mid x \in X\}$,
\item[2)] For each $1 \le i \le n$, there is a $1 \le j \le n$, such that $R_j=\{(y,x) \mid (x,y) \in R_i\}$,
\item[3)] Given $(x,z) \in R_k$, the number of $y$ such that $(x,y) \in R_i$ and $(y,z) \in R_j$ is defined to be $p_{i,j}^k$, which depends only on $i,j$ and $k$.
\end{itemize}
Then $(X,\{R_i\}_{i=0}^n)$ is called an $n$-class association scheme.

For $1 \le i \le n$, each subset $R_i$ gives rise to a digraph with vertex set $X$ in the natural way and the adjacency matrix of the digraph is denoted by $A_i$. Over the complex field, the $(n+1)$-dimensional vector space generated by $A_0, A_1, \ldots, A_n$ is indeed an algebra, which is called the Bose-Mesner algebra of the association scheme. This algebra has another basis formed by idempotent matrices $E_0, E_1, \ldots, E_n$. The transformation between these two bases can be expressed as
$$
A_i=\sum_{k=0}^nP_i(k)E_k \quad \mbox{and} \quad E_k=\frac{1}{|X|}\sum_{i=0}^nQ_{k}(i)A_i,
$$
where $P_i(k)$'s and $Q_{k}(i)$'s are called $P$-numbers and $Q$-numbers of the association scheme respectively.

Let $(X,\{R_i\}_{i=0}^n)$ be an $n$-class association scheme and $Y$ be a subset of $X$. The {\it inner distribution} of $Y$ is a sequence $(a_0,a_1,\ldots,a_n)$, where
$$
a_i:=a_i(Y)=\frac{|(Y\times Y) \cap R_i|}{|Y|}, \; 0 \le i \le n.
$$
Clearly, $a_0=1$ and $a_i$ is the average number of pairs in $Y \times Y$ belonging to $R_i$. The {\it dual inner distribution} of $Y$ is a sequence $(a_0^{\pr},a_1^{\pr},\ldots,a_n^{\pr})$, where
$$
a_k^{\pr}:=a_k^{\pr}(Y)=\sum_{i=0}^{n}a_iQ_k(i), \; 0 \le k \le n.
$$
If the set $X$ is a group, then a subset $Y$ of $X$ is {\it additive} if $Y$ is a subgroup of $X$.


\subsection{Symmetric bilinear forms over finite fields and related association schemes}

In this subsection, we introduce symmetric bilinear forms over finite fields and the related association schemes. Some crucial results concerning subsets in association scheme of symmetric bilinear forms are recorded. Let $V$ be an $n$-dimensional vector space over a finite field $\Fq$. A bilinear form $B$ on $V$ is a function from $V \times V$ to $\Fq$, satisfying:
\begin{itemize}
\item[1)] $B(x_1+x_2,y)=B(x_1,y)+B(x_2,y)$,
\item[2)] $B(\la x,y)=\la B(x,y)$,
\end{itemize}
for all $x,x_1,x_2,y \in V$ and $\la \in \Fq$. Furthermore, a bilinear form $B$ is called {\it symmetric} if $B(x,y)=B(y,x)$ for each $(x,y) \in V \times V$. Given a symmetric bilinear form $B$ on $V$, define the {\it radical} of $B$ as
$$
\Rad B:=\{y \in V \mid B(x,y)=0, \forall x \in V \}.
$$
$\Rad B$ is a vector space over $\Fq$ and the {\it rank} of symmetric bilinear form $B$ is defined by
$$
\rank(B)=n-\dim\Rad B.
$$

For the rest of this subsection, we always assume that $q$ is an odd prime power. We use $S(n,q)$ to denote the set of all symmetric bilinear forms on the vector space $V=\Fq^n$.  Let $\al_1,\al_2,\ldots,\al_n$ be a basis of $V$ over $\Fq$. Then each symmetric bilinear form $B \in S(n,q)$ corresponds to a symmetric matrix over $\Fq$:
$$
M_B=(B(\al_i,\al_j))_{i,j}.
$$
By definition, the rank of the symmetric bilinear form $B$ is equal to the rank of its associated symmetric matrix $M_B$. In particular, $B$ has rank $0$ if and only if $M_B$ is a zero matrix. Note that the correspondence between $B$ and $M_B$ depends on the specific choice of the basis. Suppose $B$ has rank $1 \le i \le n$, then there exists an $n \times n$ nonsingular matrix $L$, such that
\begin{equation}\label{eqn-stdform}
L^{T}M_BL=\begin{pmatrix}
  I_{i-1} & O & O \\
  O       & z & O \\
  O       & O & O
  \end{pmatrix}
\end{equation}
where $L^{T}$ is the transpose of $L$, $I_{i-1}$ is the identity matrix of order $i-1$, $O$ represents zero matrix with proper size and $z \in \{1,\theta\}$, in which $\theta$ is a nonsquare of $\Fq$ \cite[p. 237]{HW}. Namely, when $B$ has rank at least 1, $M_B$ is congruent to one of the two matrices with standard forms. From now on, when $q$ is odd, we use $\eta$ to denote the quadratic character of $\Fq$, with $\eta(0)=0$. For $1 \le i \le n$, $B$ has rank $r$ and type $\eta(z)$, if $M_B$ is congruent to the right hand side of (\ref{eqn-stdform}). Define $C_{i,\tau}$ to be the subset of $S(n,q)$, consisting of all symmetric bilinear forms having rank $i$ and type $\tau$. In addition, we define $C_{0,1}=C_{0,-1}$ to be the set consisting of the zero bilinear form. Hence, for $0 \le i \le n$ and $\tau \in \{1,-1\}$, we can define
$$
R_{i,\tau}=\{ (B_1,B_2) \in S(n,q)^2 \mid B_1-B_2 \in C_{i,\tau}\}.
$$
Then
$$
(S(n,q),\{R_{i,\tau} \mid 0 \le i \le n, \tau \in \{1,-1\}\})
$$
forms a $2n$-class association scheme \cite[Theorem 1]{HW}, which is denoted by $\Sym(n,q)$. Let $Y$ be a subset of $S(n,q)$. Let $(a_{i,\tau})$ be the inner distribution of $Y$, where
$$
a_{i,\tau}=\frac{|(Y\times Y) \cap R_{i,\tau}|}{|Y|}, 0 \le i \le n, \tau \in \{1,-1\}.
$$
In particular, if $Y$ is an additive subset, we have
$$
a_{i,\tau}=|Y \cap C_{i,\tau}|, 0 \le i \le n, \tau \in \{1,-1\}.
$$
$Y$ is called a {\it $d$-code} in $\Sym(n,q)$ if
$$
a_{i,1}=a_{i,-1}=0, 1 \le i \le d-1.
$$
$Y$ is called a {\it proper $d$-code} if it is a $d$-code and not a $(d+1)$-code.

For $0 \le i,k \le n$ and $\tau, \ep \in \{1,-1\}$, let $Q_{k,\ep}(i,\tau)$ be the $Q$-numbers of $\Sym(n,q)$. The dual inner distribution $(a_{k,\ep}^{\pr})$ of $Y$ is defined in the following way:
$$
a_{k,\ep}^{\pr}=\sum_{i,\tau}a_{i,\tau}Q_{k,\ep}(i,\tau).
$$
Then, $Y$ is called a {\it $t$-design} in $\Sym(n,q)$ if
$$
a_{k,1}^{\pr}=a_{k,-1}^{\pr}=0, 1 \le k \le t.
$$
Moreover, $Y$ is a $(2t+1,\ep)$-design if $Y$ is a $(2t+1)$-design and $a_{2t+2,\ep}^{\pr}=0$, where $\ep \in \{1,-1\}$. We note that the designs involved in this paper are not the usual $t$-designs studied in combinatorial design theory.

Define the $q^2$-binomial coefficient by
$$
\left[{n \atop k}\right]=\prod_{i=1}^k \frac{q^{2n-2i+2}-1}{q^{2i}-1}.
$$
A detailed treatment on $\Sym(n,q)$ was presented in \cite{Sch15}, in which the inner distribution of certain subsets of $\Sym(n,q)$ has been derived.

\begin{proposition}{\rm \cite[Theorem 3.9]{Sch15}}\label{prop-innerodd}
If $Y$ is a $(2l-1)$-code and a $(2n-2l+3)$-design in $\Sym(2n+1,q)$, then its inner distribution $(a_{i,\tau})$ satisfies
\begin{equation}
\begin{split}\label{eqn-innerdis1}
a_{2i-1,\tau}&=\frac12 \left[n \atop i-1 \right]\sum_{j=0}^{i-l}(-1)^jq^{j(j-1)}\left[i \atop j\right]\left(\frac{|Y|}{q^{(2n+1)(n+1+j-i)}}-1\right), \\
a_{2i,\tau}&=\frac12(q^{2i}+\tau\eta(-1)^iq^i)\left[n \atop i\right]\sum_{j=0}^{i-l}(-1)^jq^{j(j-1)}\left[i \atop j\right]\left(\frac{|Y|}{q^{(2n+1)(n+1+j-i)}}-1\right),
\end{split}
\end{equation}
for $i>0$.
\end{proposition}

\begin{proposition}{\rm \cite[Proposition 3.10]{Sch15}}\label{prop-innereven}
If $Y$ is a $(2l)$-code and a $(2n-2l+1)$-design in $\Sym(2n,q)$, then its inner distribution $(a_{i,\tau})$ satisfies
\begin{equation}
\begin{split}\label{eqn-innerdis2}
a_{2i-1,\tau}=&\frac12 (q^{2i}-1)\left[n \atop i\right]\sum_{j=0}^{i-l-1}(-1)^jq^{j(j-1)}\left[i-1 \atop j\right]\frac{|Y|q^{2j}}{q^{(2n+1)(n+1+j-i)}}, \\
a_{2i,\tau}=&\frac12 \left[n \atop i\right]\sum_{j=0}^{i-l}(-1)^jq^{j(j-1)}\left[i \atop j\right]\left(\frac{|Y|q^{2j}}{q^{(2n+1)(n+j-i)}}-1\right)\\
&+\frac{\tau}{2}\eta(-1)^iq^i \left[n \atop i\right]\sum_{j=0}^{i-l}(-1)^jq^{j(j-1)}\left[i \atop j\right]\left(\frac{|Y|}{q^{(2n-1)(n+j-i)}q^{2n}}-1\right)
\end{split}
\end{equation}
for $i>0$. If $Y$ is a $(2l)$-code and a $(2n-2l+1,\eta(-1)^{n-l+1})$-design in $\Sym(2n+1,q)$, then its inner distribution $(a_{i,\tau})$ satisfies
\begin{equation}
\begin{split}\label{eqn-innerdis3}
a_{2i-1,\tau}=&\frac12 \left[n \atop i-1\right]\sum_{j=0}^{i-l}(-1)^jq^{j(j-1)}\left[i \atop j\right]\left( \frac{|Y|}{q^{(2n+1)(n+1+j-i)}}-1\right)\\
&+\frac12 (-1)^{i-l}q^{(i-l)(i-l-1)}\left[n \atop l-1\right]\left(\frac{|Y|}{q^{(2n+1)(n-l+1)}}-1\right)\\
&\quad \left(\left[n-l \atop n-i+1\right](q^{n-l+1}+1)-\left[n-l+1 \atop n-i+1\right]\right), \\
a_{2i,\tau}=&\frac12 (q^{2i}+\tau\eta(-1)^iq^i) \left[n \atop i\right]\sum_{j=0}^{i-l}(-1)^jq^{j(j-1)}\left[i \atop j\right]\left(\frac{|Y|}{q^{(2n+1)(n+1+j-i)}}-1\right)\\
&+\frac{1}{2}(-1)^{i-l}q^{(i-l+1)(i-l)} \left[n \atop l-1\right]\left[n-l \atop n-i\right](q^{n-l+1}+1)\left(\frac{|Y|}{q^{(2n+1)(n-l+1)}}-1\right)
\end{split}
\end{equation}
for $i>0$.
\end{proposition}

The above two propositions play important roles in our analysis of the BCH code $\codespe$, when $q$ is odd.

\subsection{Alternating bilinear forms over finite fields and related association schemes}

In this subsection, we introduce alternating bilinear forms over finite fields and related association schemes. A result concerning subsets in the association scheme of alternating bilinear forms is also recalled. Let $V$ be an $n$-dimensional vector space over $\Fq$. A bilinear form $B$ on $V$ is called {\it alternating} if $B(x,x)=0$ for each $x \in V$. We use $A(n,q)$ to denote the set of all alternating bilinear forms over $V=\Fq^n$. Note that if the ground finite field $\Fq$ has even characteristic, an alternating bilinear form is necessarily symmetric. Indeed, since $B$ is alternating, then we have $B(x+y,x+y)=0$ for each $x,y \in V$. Thus, we have $B(x,y)+B(y,x)=0$ for each $x,y \in V$, which implies $B$ is symmetric when $\Fq$ has even characteristic. In the rest of this subsection, we always assume that $q$ is an even prime power.

Let $\al_1,\al_2,\ldots,\al_n$ be a basis of $V$ over $\Fq$. Then each alternating bilinear form $B \in A(n,q)$ corresponds to a skew-symmetric matrix over $\Fq$:
$$
M_B=(B(\al_i,\al_j))_{i,j}.
$$
Since $B$ is also symmetric, its radical and rank have been be defined in the last subsection. By definition, the rank of the alternating bilinear form $B$ is equal to the rank of its associated skew-symmetric matrix $M_B$, which is necessarily even.
For $0 \le i \le \lhn$, define $C_{2i}$ to be the subset of $A(n,q)$, consisting of alternating bilinear forms having rank $2i$. Hence, for $0 \le i \le \lhn$, we can define
$$
R_{2i}=\{(B_1,B_2) \in A(n,q)^2 \mid B_1-B_2 \in C_{2i}\}.
$$
Then,
$$
(A(n,q),\{R_{2i} \mid 0 \le i \le \lhn\})
$$
forms a $\lhn$-class association scheme \cite[p. 28]{DG}, denoted by $\Alt(n,q)$. Let $Y$ be a subset of $A(n,q)$. Let $(b_{2i})$ be the inner distribution of $Y$, where
$$
b_{2i}=\frac{|(Y\times Y) \cap R_{2i}|}{|Y|}, \; 0 \le i \le \lhn.
$$
In particular, if $Y$ is an additive subset,
$$
b_{2i}=|Y \cap C_{2i}|, \; 0 \le i \le \lhn.
$$
$Y$ is called a {\it $2d$-code} in $\Alt(n,q)$ if
$$
b_{2i}=0, \; 1 \le i \le d-1.
$$
$Y$ is called a {\it proper $2d$-code} if it is a $2d$-code and not a $(2d+2)$-code. The following proposition gives an upper bound on the size of $2d$-code in $\Alt(n,q)$.

\begin{proposition}{\rm \cite[Theorem 4]{DG}}\label{prop-bound}
For $0 \le d \le \lhn$, let $Y$ be a $2d$-code in $\Alt(n,q)$. Then,
$$
|Y|\le\begin{cases}
  q^{n(\frac{n+1}{2}-d)} & \mbox{if $n$ is odd,} \\
  q^{(n-1)(\frac{n+2}{2}-d)} & \mbox{if $n$ is even.}
\end{cases}
$$
\end{proposition}

The above proposition plays a crucial role in our analysis of the BCH code $\codespe$, when $q$ is even.

\subsection{Quadratic forms over finite fields and related association schemes}\label{sec3D}

Let $V$ be an $n$-dimensional vector space over $\Fq$. A quadratic form $Q$ on $V$ is defined to be a function from $V$ to $\Fq$, satisfying:
\begin{itemize}
\item[1)] $Q(\la x)=\la^2Q(x)$,
\item[2)] When $q$ is odd, $Q(x+y)=Q(x)+Q(y)+2B_Q(x,y)$,
\item[3)] When $q$ is even, $Q(x+y)=Q(x)+Q(y)+B_Q(x,y)$,
\end{itemize}
for all $x,y \in V$ and $\la \in \Fq$, where $B_Q$ is a symmetric bilinear form on $V$ associated with $Q$.

\begin{remark}
Usually, the bilinear form $B_Q$ associated with the quadratic form $Q$ is introduced by $Q(x+y)=Q(x)+Q(y)+B_Q(x,y)$, whenever $q$ is odd or even. The definition above makes a modification when $q$ is odd. On one hand, this modification brings no essential difference since $2$ has a multiplicative inverse in $\Fq$ when $q$ is odd. On the other hand, this modification makes it easier to show that $B_Q$ has the same rank and type as $Q$, which will be demonstrated in Proposition~\ref{prop-innerdis}.
\end{remark}

\begin{remark}\label{rem-quadbili}
We remark that the relation between $Q$ and the associated symmetric bilinear form $B_Q$ depends heavily on $q$ being odd or even. When $q$ is odd, we have $Q(x)=B_Q(x,x)$, for all $x \in V$. Hence, $Q$ can be uniquely recovered from $B_Q$, which means $Q$ and $B_Q$ essentially carry the same information. When $q$ is even, we have $B_Q(x,x)=0$ for all $x \in V$. Therefore, the bilinear form $B_Q$ is not only symmetric, but also alternating. Define $\cH$ to be the set $\{ Q \in A(n,q) \mid Q(x+y)=Q(x)+Q(y), \forall x,y \in V\}$. Then each $B_Q$ corresponds to a coset $Q+\cH$, which means we cannot recover $Q$ from $B_Q$.
\end{remark}

Define the {\it radical} of $Q$ as
$$
\Rad Q = Q^{-1}(0) \cap \Rad B_Q.
$$
$\Rad Q$ is a vector space over $\Fq$ and the {\it rank} of quadratic form $Q$ is
$$
\rank(Q)=n-\dim \Rad Q.
$$
We are going to present the relation between $\rank(Q)$ and $\rank(B_Q)$. Remark~\ref{rem-quadbili} suggests this relation is quite different when $q$ is odd or even.

\begin{lemma}\label{lem-rank}
Let $Q$ be a quadratic form and $B_Q$ be the associated bilinear form. Then we have
\begin{equation*}
\begin{cases}
  \rank(Q)=\rank(B_Q) & \mbox{if $q$ is odd,} \\
  \rank(B_Q) \le \rank(Q) \le \rank(B_Q)+1 & \mbox{if $q$ is even.}
\end{cases}
\end{equation*}
\end{lemma}
\begin{proof}
If $q$ is odd, for each $x \in \Rad B_Q$, we have $Q(x)=B_Q(x,x)=0$. Thus, $\Rad B_Q \subset Q^{-1}(0)$ and $\Rad Q=\Rad B_Q$. This implies $\rank(Q)=\rank(B_Q)$ when $q$ is odd.

If $q$ is even, we consider the mapping
\begin{align*}
Q|_{\Rad B_Q}: \Rad B_Q & \rightarrow \Fq \\
                x   & \rightarrow Q(x)
\end{align*}
which is a group homomorphism. Note that if $x \in \Rad B_Q$, then $\la x \in \Rad B_Q$ and $Q(\la x)=\la^2Q(x)$ for each $\la \in \Fq$. Then $Q|_{\Rad B_Q}$ is either a zero or a surjective homomorphism. If $Q|_{\Rad B_Q}$ is zero, then $\Rad B_Q \subset Q^{-1}(0)$, which implies $\Rad Q = \Rad B_Q$ and $\rank(Q)=\rank(B_Q)$. If $Q|_{\Rad B_Q}$ is surjective, then $\Rad B_Q/(Q^{-1}(0) \cap \Rad B_Q)$ is a one-dimensional vector space over $\Fq$, which means $\dim \Rad Q = \dim\Rad B_Q-1$ and $\rank(Q)=\rank(B_Q)+1$.
\end{proof}

Let $Q$ be a quadratic form on $V$ and let $x=(x_1,x_2,\ldots,x_n) \in V$ be a column vector. When $q$ is odd, $Q$ can be uniquely expressed in the following form:
$$
Q(x)=\sum_{i,j=1}^n c_{ij}x_ix_j,
$$
with $c_{ij} \in \Fq$ and $c_{ij}=c_{ji}$, $1 \le i,j \le n$. In this case, $Q$ corresponds to an $n \times n$ symmetric matrix $C=(c_{ij})$ over $\Fq$. When $q$ is even, $Q$ can be uniquely expressed in the following form:
$$
Q(x)=\sum_{1 \le i \le j \le n} c_{ij}x_ix_j,
$$
with $c_{ij} \in \Fq$ for $1 \le i,j \le n$. In this case, $Q$ corresponds to an $n \times n$ upper triangular matrix $C=(c_{ij})$ over $\Fq$. Therefore, given a quadratic from $Q$, we call the corresponding matrix $C$ as the {\it coefficient matrix} of $Q$ and we have
$$
Q(x)=x^{T}Cx, \forall x \in V,
$$
where $x^{T}$ is the transpose of the column vector $x$.

\begin{definition}
Let $V$ be an $n$-dimensional vector space over $\Fq$. Let $Q$ and $Q^{\pr}$ be two quadratic forms on $V$, with coefficient matrix $C$ and $C^{\pr}$, respectively. Then $Q$ and $Q^{\pr}$ are equivalent, if there is an $n \times n$ nonsingular matrix $D$ over $\Fq$, such that $C^{\pr}=D^{T}CD$, where $D^{T}$ is the transpose of $D$.
\end{definition}

Note that for $x=(x_1,x_2,\ldots,x_n) \in V$, $Q(x)=x^TCx$ and $Q^{\pr}(x)=x^TC^{\pr}x=(Dx)^{T}C(Dx)$. Thus, $Q$ and $Q^{\pr}$ are equivalent if and only if $Q^{\pr}$ can be obtained from $Q$, by applying a nonsingular linear transformation on the variables $x_1,x_2,\ldots,x_n$. The next proposition says each quadratic form over finite field is equivalent to one of the following canonical forms.

\begin{proposition}{\rm \cite[Theorems 6.21 and 6.30]{LN}}\label{prop-equivalence}
Let $Q$ be a quadratic form on an $n$-dimensional vector space $V$ over $\Fq$.
\begin{itemize}
\item[1)] Let $q$ be an odd prime power. Then $Q$ is equivalent to the following quadratic form
$$
\sum_{j=1}^r a_jx_j^2, \quad \mbox{with $a_j \in \Fq^*$ and $\eta(\prod_{j=1}^r a_j)=\tau$},
$$
where $0 \le r \le n$ and $\tau \in \{1,-1\}$. In this case, we say that $Q$ has rank $r$ and type $\tau$.
\item[2)] Let $q$ be an even prime power. When $Q$ has odd rank, $Q$ is equivalent to
$$
\sum_{j=1}^r x_{2j-1}x_{2j}+x_{2r+1}^2.
$$
In this case, we say that $Q$ has rank $2r+1$ and type $1$. When $Q$ has even rank, $Q$ is either equivalent to
$$
\sum_{j=1}^r x_{2j-1}x_{2j},
$$
or equivalent to
$$
\sum_{j=1}^{r} x_{2j-1}x_{2j}+x_{2r-1}^2+\la x_{2r}^2, \Tr^q_2(\la)=1.
$$
We say $Q$ has rank $2r$ and type $0$ in the former case and has rank $2r$ and type $2$ in the latter one.
\end{itemize}
\end{proposition}

Define a function $\ups$ on $\Fq$ as
$$
\ups(x)=\begin{cases}
  -1 & \mbox{if $x \in \Fq^*$,} \\
  q-1 & \mbox{if $x=0$.}
\end{cases}
$$
The following two propositions concern the number of solutions to quadratic equations over finite fields, which are direct consequences of \cite[Theorems 6.26, 6.27 and 6.32]{LN}.

\begin{proposition}\label{prop-numsoluodd}
Let $q$ be an odd prime power and $Q$ be a quadratic form on $\Fq^n$ with rank $r$ and type $\tau$. Then for $h \in \Fq$, the number of solutions to the equation $Q(x)=h$ is
$$
N(h)=\begin{cases}
  q^{n-1}+\tau\eta(-1)^{\frac{r-1}{2}}\eta(h)q^{n-\frac{r+1}{2}} & \mbox{if $r$ is odd,} \\
  q^{n-1}+\tau\eta(-1)^{\frac{r}{2}}\ups(h)q^{n-\frac{r+2}{2}} & \mbox{if $r$ is even.} \\
\end{cases}
$$
\end{proposition}

\begin{proposition}\label{prop-numsolueven}
Let $q$ be an even prime power and $Q$ be a quadratic form on $\Fq^n$ with rank $r$ and type $\tau$. Then for $h \in \Fq$, the number of solutions to the equation $Q(x)=h$ is
$$
N(h)=\begin{cases}
  q^{n-1}+\ups(h)q^{n-\frac{r+2}{2}} & \mbox{if $r$ is even, $\tau=0$,} \\
  q^{n-1}                            & \mbox{if $r$ is odd, $\tau=1$,} \\
  q^{n-1}-\ups(h)q^{n-\frac{r+2}{2}} & \mbox{if $r$ is even, $\tau=2$,} \\
\end{cases}
$$
\end{proposition}

Let $V$ be an $n$-dimensional vector space over $\Fq$. We use $Q(n,q)$ to denote the set of all quadratic forms on $V$. Suppose $q$ is odd. For $1 \le i \le n$ and $\tau \in \{1,-1\}$, define $D_{i,\tau}$ to be the subset of $Q(n,q)$ consisting of all quadratic forms with rank $i$ and type $\tau$. Define $D_{0,1}=D_{0,-1}$ to be the set consisting of the zero quadratic form. For $0 \le i \le n$ and $\tau \in \{1,-1\}$,
$$
S_{i,\tau}=\{(Q_1,Q_2) \in Q(n,q) \mid Q_1-Q_2 \in D_{i,\tau}\}.
$$
Then
$$
(Q(n,q),\{S_{i,\tau} \mid 0\le i\le n, \tau \in \{1,-1\}\})
$$
forms a $2n$-class association scheme \cite{WWMM}, denoted by $\Qua(n,q)$. Let $Y$ be a subset of $Q(n,q)$. Then the inner distribution $(c_{i,\tau})$ of $Y$ is defined by
$$
c_{i,\tau}=\frac{|(Y \times Y) \cap S_{i,\tau}|}{|Y|}, 0 \le i \le n, \tau\in\{1,-1\}.
$$
If $Y$ is an additive subset, we have
$$
c_{i,\tau}=|Y \cap D_{i,\tau}|, 0 \le i \le n, \tau\in\{1,-1\}.
$$

Suppose $q$ is even. For $0 \le i \le \lhn$, $\tau \in \{0,1,2\}$ and $0 \le 2i+\tau \le n$, define $E_{2i+\tau,\tau}$ to be the subset of $Q(n,q)$ consisting of all quadratic forms with rank $2i+\tau$ and type $\tau$. For $0 \le i \le \lhn$, $\tau \in \{0,1,2\}$ and $0 \le 2i+\tau \le n$,
$$
T_{2i+\tau,\tau}=\{(Q_1,Q_2) \in Q(n,q) \mid Q_1-Q_2 \in E_{2i+\tau,\tau}\}.
$$
Then
$$
(Q(n,q),\{T_{2i+\tau,\tau} \mid 0 \le i \le \lhn, \tau \in \{0,1,2\}, 0\le 2i+\tau \le n \})
$$
forms an $n+\lhn$-class association scheme \cite{FWMM,WWMM}, denoted by $\Qua(n,q)$. Let $Y$ be a subset of $Q(n,q)$. Then the inner distribution $(d_{2i+\tau,\tau})$ of $Y$ is defined by
$$
d_{2i+\tau,\tau}=\frac{|(Y \times Y) \cap T_{2i+\tau,\tau}|}{|Y|}, 0 \le i \le \lhn, \tau \in \{0,1,2\}, 0\le 2i+\tau \le n.
$$
If $Y$ is an additive subset, we have
$$
d_{2i+\tau,\tau}=|Y \cap E_{2i+\tau,\tau}|, 0 \le i \le \lhn, \tau \in \{0,1,2\}, 0\le 2i+\tau \le n.
$$

Let $\cQ \subset Q(n,q)$ be a set of quadratic forms. When $q$ is odd, define $\cS=\{B_Q \mid Q \in \cQ\}$, where $B_Q$ is the associated symmetric bilinear form. Then we have $\cS \subset \Sym(n,q)$. When $q$ is even, define $\cA=\{B_Q \mid Q \in \cQ\}$, where $B_Q$ is the associated alternating bilinear form. Then we have $\cA \subset \Alt(n,q)$. The following proposition suggests that the inner distribution of $\cQ$ in $\Qua(n,q)$ has close connection with the inner distribution of $\cS$ in $\Sym(n,q)$ when $q$ is odd and with the inner distribution of $\cA$ in $\Alt(n,q)$ when $q$ is even.

\begin{proposition}\label{prop-innerdis}
\begin{itemize}
\item[1)] When $q$ is odd, let $(c_{i,\tau})$ and $(a_{i,\tau})$ be the inner distributions of $\cQ$ and $\cS$ respectively. Then we have $c_{i,\tau}=a_{i,\tau}$ for $0 \le i \le n$ and $\tau \in \{1,-1\}$.
\item[2)] When $q$ is even, let $(d_{2i+\tau,\tau})$ and $(b_{2i})$ be the inner distributions of $\cQ$ and $\cA$ respectively. Then we have $d_{2i,0}+d_{2i+1,1}+d_{2i,2}=b_{2i}$ for $0 \le i \le \frac{n-1}{2}$.
\end{itemize}
\end{proposition}
\begin{proof}
1) Let $Q \in Q(n,q)$. When $q$ is odd, by Lemma~\ref{lem-rank}, we have $\rank(Q)=\rank(B_Q)$ for each $Q \in \cQ$. Recall that $B(x,y)=\frac12 (Q(x+y)-Q(x)-Q(y))$ for all $x,y \in V$. It can be easily verified that the coefficient matrix of the quadratic form $Q$ coincides with the symmetric matrix associated to $B_Q$. By the definitions of the type of $Q$ and $B_Q$, it is easy to see that $Q$ and $B_Q$ have the same type. Hence, each pair of $Q$ and $B_Q$ have the same rank and type. Therefore, $c_{i,\tau}=a_{i,\tau}$ for $0 \le i \le n$ and $\tau \in \{1,-1\}$.

2) Let $Q \in Q(n,q)$. When $q$ is even, by Lemma~\ref{lem-rank}, we have $\rank(B_Q) \le \rank(Q) \le \rank(B_Q)+1$. Suppose $B_Q$ has rank $2i$, then $Q$ must have rank $2i$, type $0$, or rank $2i+1$, type $1$ or rank $2i$, type $2$. Therefore, $d_{2i,0}+d_{2i+1,1}+d_{2i,2}=b_{2i}$ for $0 \le i \le \frac{n-1}{2}$.
\end{proof}

The above proposition says when $q$ is odd, the inner distribution of $\cQ$ can be derived from that of $\cS$. And when $q$ is even, some information on the inner distribution of $\cQ$ can be derived from that of $\cA$. This fact plays a fundamental role in the determination of the minimum distance of $\codespe$.

Below, most of our computation is carried out in the finite field $\Fqm$, rather than the vector space $\Fq^m$. For the sake of convenience, we say a quadratic form, a symmetric bilinear form or an alternating bilinear form is defined on the finite field $\Fqm$, if it is a quadratic form, a symmetric bilinear form or an alternating bilinear form on the vector space $\Fq^m$. In some places, we simply use the notation $\Fqm$ to represent the additive group of this finite field, which forms a vector space $\Fq^m$.

\section{Determining the minimum distance when $q$ is odd}\label{sec4}

For $\hmmt \le i \le m-\ltm-1$, the dimension and Bose distance of $\codespe$ have been settled in Proposition~\ref{prop-traceprop}. In this section, we are going to show that the minimum distance of $\codespe$ is equal to its Bose distance $\de_i$, when $q$ is an odd prime power. Throughout the rest of this section, we always assume that $q$ is odd.

By Proposition~\ref{prop-cosetunion}, in order to determine the minimum distance of $\codespe$, it suffices to compute the minimum weight of the subcode $(Q(x))_{x \in \Fqm^*}+\PRM_q(1,m)$ for each $Q \in \cQ_1$ when $m$ is odd and for each $Q \in \cQ_2$ when $m$ is even. The following proposition shows when $q$ is odd, the weight distribution of the subcode $(Q(x))_{x \in \Fqm^*}+\PRM_q(1,m)$ depends only on the rank and the type of the quadratic form $Q$ on $\Fqm$.

As a preparation, consider a function $f$ from $\Fqm$ to $\Fq$. Define $N(f)$ to be the number of $x \in \Fqm$ such that $f(x)=0$. Define $\wt(f)$ to be the number of $x \in \Fqm^*$ such that $f(x)\ne0$. Then we have
$$
\wt(f)=\begin{cases}
  q^m-N(f) & \mbox{if $f(0)=0$,} \\
  q^m-1-N(f) & \mbox{if $f(0) \ne 0$.}
\end{cases}
$$
For $h \in \Fq$, define
$$
\De_{0,h}=\begin{cases}
  1 & \mbox{if $h=0$,} \\
  0 & \mbox{if $h \ne 0$.}
\end{cases}
$$
Then we have
\begin{equation}\label{eqn-fweight}
\wt(f)=q^m-1-N(f)+\De_{0,f(0)}.
\end{equation}

\begin{proposition}\label{prop-Qweightodd}
Let $q$ be odd and $Q$ be a quadratic form of rank $r \ge 1$ and type $\tau$ on $\Fqm$. Then the weight enumerator of $(Q(x))_{x \in \Fqm^*}+\PRM_q(1,m)$ is
\begin{align*}
U_{r,\tau}(Z)=&\frac{(q-1)}{2}((q-1)q^{r-1}-\tau\eta(-1)^{\frac{r-1}{2}}q^{\frac{r-1}{2}})Z^{q^m-q^{m-1}-\tau\eta(-1)^{\frac{r-1}{2}}q^{m-\frac{r+1}{2}}-1}\\
&+\frac{(q-1)}{2}(q^{r-1}+\tau\eta(-1)^{\frac{r-1}{2}}q^{\frac{r-1}{2}})Z^{q^m-q^{m-1}-\tau\eta(-1)^{\frac{r-1}{2}}q^{m-\frac{r+1}{2}}}\\
&+(q-1)(q^m-q^{r}+q^{r-1})Z^{q^m-q^{m-1}-1}+(q^m-q^{r}+q^{r-1})Z^{q^m-q^{m-1}}\\
&+\frac{(q-1)}{2}((q-1)q^{r-1}+\tau\eta(-1)^{\frac{r-1}{2}}q^{\frac{r-1}{2}})Z^{q^m-q^{m-1}+\tau\eta(-1)^{\frac{r-1}{2}}q^{m-\frac{r+1}{2}}-1}\\
&+\frac{(q-1)}{2}(q^{r-1}-\tau\eta(-1)^{\frac{r-1}{2}}q^{\frac{r-1}{2}})Z^{q^m-q^{m-1}+\tau\eta(-1)^{\frac{r-1}{2}}q^{m-\frac{r+1}{2}}}
\end{align*}
when $r$ is odd and
\begin{align*}
U_{r,\tau}(Z)=&(q-1)(q^{r-1}-\tau\eta(-1)^{\frac{r}{2}}q^{\frac{r-2}{2}})Z^{q^m-q^{m-1}-\tau\eta(-1)^{\frac{r}{2}}q^{m-\frac{r+2}{2}}(q-1)-1}\\
&+(q^{r-1}+\tau\eta(-1)^{\frac{r}{2}}q^{\frac{r-2}{2}}(q-1))Z^{q^m-q^{m-1}-\tau\eta(-1)^{\frac{r}{2}}q^{m-\frac{r+2}{2}}(q-1)}\\
&+(q-1)(q^m-q^{r})Z^{q^m-q^{m-1}-1}+(q^m-q^{r})Z^{q^m-q^{m-1}}\\
&+(q-1)((q-1)q^{r-1}+\tau\eta(-1)^{\frac{r}{2}}q^{\frac{r-2}{2}})Z^{q^m-q^{m-1}+\tau\eta(-1)^{\frac{r}{2}}q^{m-\frac{r+2}{2}}-1}\\
&+(q-1)(q^{r-1}-\tau\eta(-1)^{\frac{r}{2}}q^{\frac{r-2}{2}})Z^{q^m-q^{m-1}+\tau\eta(-1)^{\frac{r}{2}}q^{m-\frac{r+2}{2}}}
\end{align*}
when $r$ is even. The above weight enumerators are also summarized in Tables \ref{tab-qoddrodd} and \ref{tab-qoddreven}.
\end{proposition}

\begin{table}
\begin{center}
\ra{1.5}\caption{Weight distribution of $(Q(x))_{x \in \Fqm^*}+\PRM_q(1,m)$, $q$ and $r$ odd}
\begin{tabular}{|c|c|}
\hline
Weight  &   Frequency  \\ \hline
$q^m-q^{m-1}\pm\tau\eta(-1)^{\frac{r-1}{2}}q^{m-\frac{r+1}{2}}-1$ & $\frac{(q-1)}{2}((q-1)q^{r-1}\pm\tau\eta(-1)^{\frac{r-1}{2}}q^{\frac{r-1}{2}})$ \\ \hline
$q^m-q^{m-1}\pm\tau\eta(-1)^{\frac{r-1}{2}}q^{m-\frac{r+1}{2}}$ & $\frac{(q-1)}{2}(q^{r-1}\mp\tau\eta(-1)^{\frac{r-1}{2}}q^{\frac{r-1}{2}})$ \\ \hline
$q^m-q^{m-1}-1$ & $(q-1)(q^m-q^{r}+q^{r-1})$ \\ \hline
$q^m-q^{m-1}$ & $q^m-q^{r}+q^{r-1}$ \\ \hline
\end{tabular}
\end{center}
\label{tab-qoddrodd}
\end{table}

\begin{table}
\begin{center}
\ra{1.5}\caption{Weight distribution of $(Q(x))_{x \in \Fqm^*}+\PRM_q(1,m)$, $q$ odd and $r$ even}
\begin{tabular}{|c|c|}
\hline
Weight  &   Frequency  \\ \hline
$q^m-q^{m-1}-\tau\eta(-1)^{\frac{r}{2}}q^{m-\frac{r+2}{2}}(q-1)-1$ & $(q-1)(q^{r-1}-\tau\eta(-1)^{\frac{r}{2}}q^{\frac{r-2}{2}})$ \\ \hline
$q^m-q^{m-1}-\tau\eta(-1)^{\frac{r}{2}}q^{m-\frac{r+2}{2}}(q-1)$ & $q^{r-1}+\tau\eta(-1)^{\frac{r}{2}}q^{\frac{r-2}{2}}(q-1)$ \\ \hline
$q^m-q^{m-1}-1$ & $(q-1)(q^m-q^{r})$ \\ \hline
$q^m-q^{m-1}$ & $q^m-q^{r}$ \\ \hline
$q^m-q^{m-1}+\tau\eta(-1)^{\frac{r}{2}}q^{m-\frac{r+2}{2}}-1$ & $(q-1)((q-1)q^{r-1}+\tau\eta(-1)^{\frac{r}{2}}q^{\frac{r-2}{2}})$ \\ \hline
$q^m-q^{m-1}+\tau\eta(-1)^{\frac{r}{2}}q^{m-\frac{r+2}{2}}$ & $(q-1)(q^{r-1}-\tau\eta(-1)^{\frac{r}{2}}q^{\frac{r-2}{2}})$ \\ \hline
\end{tabular}
\end{center}
\label{tab-qoddreven}
\end{table}
\begin{proof}
Let $\cL$ be the set of all homogenous linear functions on $\Fqm$. Then, the weight enumerator of $(Q(x))_{x \in \Fqm^*}+\PRM_q(1,m)$ can be read from the multiset $\{\wt(Q+L+c) \mid L \in \cL, c \in \Fq\}$. The conclusion follows from (\ref{eqn-fweight}) and Lemma~\ref{lem-quadoddset}.
\end{proof}

Note that $\cQ_1$ and $\cQ_2$ defined in (\ref{eqn-Q1}) and (\ref{eqn-Q2}) are additive subsets of $Q(m,q)$. By Propositions~\ref{prop-cosetunion} and \ref{prop-Qweightodd}, we can obtain the weight distribution of $\codespe$, if the inner distributions of $\cQ_1$ and $\cQ_2$ are known. For $\cQ_1$ and $\cQ_2$, define the corresponding sets of symmetric bilinear forms as
$$
\cS_j=\{B_Q \mid Q \in \cQ_j\},
$$
where $B_Q(x,y)=\frac12 (Q(x+y)-Q(x)-Q(y))$ and $j=1,2$. More precisely, we have
\begin{equation}\label{eqn-S1}
\begin{aligned}
  \cS_1:=\cS_1(i)&=\Bigg\{\Tqmq\bigg(\sum_{j=\hmpo}^{i+1}\Big(\frac{\la_j}{2}x^{q^j}+\big(\frac{\la_j}{2}\big)^{q^{-j}}x^{q^{-j}}\Big)y\bigg) \mid \la_{\hmpo},\ldots,\la_{i+1} \in \Fqm \Bigg\} \\
                 &=\Bigg\{\Tqmq\Big(\sum_{j=\hmpo}^{i+1}(\la_jx^{q^j}+\la_j^{q^{-j}}x^{q^{-j}})y\Big) \mid \la_{\hmpo},\ldots,\la_{i+1} \in \Fqm \Bigg\}
\end{aligned}
\end{equation}
and
\begin{equation}\label{eqn-S2}
\begin{aligned}
  \cS_2:=\cS_2(i)=\Bigg\{\Tqmq\bigg(\Big(\frac{\la_{\hm}}{2}x^{q^{\hm}}+\sum_{j=\hmpt}^{i+1}&\big(\frac{\la_j}{2}x^{q^j}+\big(\frac{\la_j}{2}\big)^{q^{-j}}x^{q^{-j}}\big)\Big)y\bigg)\\ &\Big | \la_{\hm} \in \Fqhm, \la_{\hmpt},\ldots,\la_{i+1} \in \Fqm \Bigg\} \\
                 =\Bigg\{\Tqmq\Big(\big(\la_{\hm}x^{q^{\hm}}+\sum_{j=\hmpt}^{i+1}&(\la_jx^{q^j}+\la_j^{q^{-j}}x^{q^{-j}})\big)y\Big) \\
                  &\Big | \la_{\hm} \in \Fqhm, \la_{\hmpt},\ldots,\la_{i+1} \in \Fqm \Bigg\}
\end{aligned}
\end{equation}
with $|\cS_1|=|\cS_2|=q^{m(i-\frac{m-3}{2})}$. By Proposition~\ref{prop-innerdis}, for $j=1,2$, $\cQ_j$ has the same inner distribution with $\cS_j$. To derive the inner distributions of $\cS_1$ and $\cS_2$, we are going to study them in the context of association schemes. As a preparation, we state the following proposition which will be used to prove a subset of $\Sym(m,q)$ is a $t$-design.

\begin{proposition}{\rm \cite[Theorem 3.11]{Sch15}}\label{prop-tdesign}
Let $U$ be a $t$-dimensional subspace of $\Fqm$ and let $A$ be a symmetric bilinear form on $U$. Then a subset $Y$ of $\Sym(m,q)$ is a $t$-design if and only if the number of forms in $Y$ that are an extension of $A$ is a constant, which is independent of the choice of $U$ and $A$.
\end{proposition}

The above proposition provides a combinatorial characterization of $t$-designs in $\Sym(m,q)$. In order to exploit this characterization, we need the following two lemmas. Let $V$ be an $m$-dimensional vector space over $\Fq$. Denote the set of all bilinear forms on $V$ as $B(m,q)$. Let $C \in B(m,q)$. Define an associated bilinear form $C^{\pr}$ on $V$ as $C^{\pr}(x,y)=C(y,x)$, in which $x,y \in V$.

\begin{lemma}\label{lem-bilisym}
Let $J$ be a multiset in which each element of $B(m,q)$ occurs a constant number of times. Then, each element of $S(m,q)$ occurs a constant number of times in the multiset
$$
\{C+C^{\pr} \mid C \in J\}.
$$
\end{lemma}
\begin{proof}
Set $x=(x_1,x_2,\ldots,x_m) \in V$ and $y=(y_1,y_2,\ldots,y_m) \in V$. Note that each $C\in B(m,q)$ can be uniquely represented as $C(x,y)=\sum_{1 \le i,j \le m} c_{ij}x_iy_j$, where $c_{ij} \in \Fq$, $1 \le i,j \le m$. Consequently, $C^{\pr}$ can be represented as $C^{\pr}(x,y)=\sum_{1 \le i,j \le m} c_{ji}x_iy_j$. Besides, each $A\in S(m,q)$ can be uniquely represented as $A(x,y)=\sum_{1 \le i,j \le m} a_{ij}x_iy_j$, where $a_{ij} \in \Fq$, $a_{ij}=a_{ji}$, $1 \le i,j \le m$. Hence, $A=C+C^{\pr}$ if and only if $a_{ij}=c_{ij}+c_{ji}$ for $1 \le i,j \le m$. Therefore, the number of $C \in B(m,q)$ satisfying $A=C+C^{\pr}$ is a constant independent of the choice of $A$. That means each element of $S(m,q)$ occurs a constant number of times in the multiset $\{C+C^{\pr} \mid C \in B(m,q)\}$. Since $J$ is a multiset in which each element of $B(m,q)$ has the same multiplicity, then each element of $S(m,q)$ occurs a constant number of times in the multiset $\{C+C^{\pr} \mid C \in J\}$.
\end{proof}

\begin{lemma}{\rm \cite[Lemma 4.6]{Sch15}}\label{lem-bilirep}
Let $U$ be a $t$-dimensional subspace of $\Fqm$. Let $l$ be an integer and let $s$ be an integer coprime to $m$. Then every bilinear form from $U \times \Fqm$ to $\Fq$ can be uniquely expressed in the form
$$
B_\la(x,y)=\Tr^{q^m}_q\big(\sum_{j=0}^{t-1}\la_jx^{q^{s(j-l)}}y\big),
$$
where $\la=(\la_0,\la_1,\ldots,\la_{t-1}) \in \Fqm^t$.
\end{lemma}

Let $U$ be a subspace of $\Fqm$ and $B$ be a bilinear form over $\Fqm$. Let $B|_{U}$ denote the restriction of $B$ on $U \times U$ and $B|_{U\times\Fqm}$ denote the restriction of $B$ on $U\times\Fqm$. Now, we are ready to show that certain subsets of $\Sym(m,q)$ are $t$-designs.

\begin{proposition}\label{prop-tdesignfamily}
Let $s$ be an integer coprime to $m$.
\begin{itemize}
\item[1)] When $m$ is odd and $t$ is even with $2 \le t < m$, define a set of symmetric bilinear forms on $\Fqm$ as
$$
\cY_1:=\Big\{\Tr^{q^m}_q\big(\sum_{j=\frac{m+1}{2}}^{\frac{m+t-1}{2}}\la_j(x^{q^{sj}}y+xy^{q^{sj}})\big) \mid \la_{\frac{m+1}{2}},\ldots,\la_{\frac{m+t-1}{2}} \in \Fqm\Big\}.
$$
Then $\cY_1$ is a $t$-design in $\Sym(m,q)$.
\item[2)] When $m$ is even and $t$ is odd with $1 \le t < m$, define a set of symmetric bilinear forms on $\Fqm$ as
\begin{align*}
\cY_2:=\Big\{\Tr^{q^m}_q\big(\la_{\hm}x^{q^{\frac{sm}{2}}}y+\sum_{j=\frac{m+2}{2}}^{\frac{m+t-1}{2}}&\la_j(x^{q^{sj}}y+xy^{q^{sj}})\big) \\ &\mid \la_{\hm} \in \Fqhm, \la_{\frac{m+2}{2}},\ldots,\la_{\frac{m+t-1}{2}} \in \Fqm\Big\}.
\end{align*}
Then $\cY_2$ is a $t$-design in $\Sym(m,q)$.
\end{itemize}
\end{proposition}
\begin{proof}
Let $U$ be a $t$-dimensional subspace of $\Fqm$ and $A$ be a symmetric bilinear form on $U$. Set $l=\frac{m-t+1}{2}$. For $\mu=(\mu_0,\mu_1,\cdots,\mu_{t-1})\in\Fqm^t$, define a bilinear form $B_{\mu}$ on $\Fqm$ as
$$
B_{\mu}(x,y)=\Tqmq(\sum_{j=0}^{t-1}\mu_jx^{q^{s(j+l)}}y)=\Tqmq(\sum_{j=l}^{\frac{m+t-1}{2}}\mu_{j-l}x^{q^{sj}}y)
$$
where $x, y \in \Fqm$. By Lemma~\ref{lem-bilirep}, when $\mu$ ranges over $\Fqm^t$, $B_{\mu}|_{U\times\Fqm}$ ranges over all bilinear forms over $U \times \Fqm$ exactly once. Then $\{B_{\mu}|_{U} \mid \mu \in \Fqm^t\}$ is a multiset in which each bilinear form on $U$ occurs a constant number of times, which depends only on $t$. By Lemma~\ref{lem-bilisym}, $A$ occurs a constant number of times in the multiset
$$
\{C+C^{\pr} \mid C \in \{B_{\mu}|_U \mid \mu \in \Fqm^t \}\}.
$$
Equivalently, the number of elements in the multiset $\{D+D^{\pr} \mid D \in \{B_{\mu} \mid \mu \in \Fqm^t\}\}$, which is an extension of $A$, is a constant independent of $U$ and $A$.

Let $m$ be odd. Given $\mu=(\mu_0,\mu_1,\cdots,\mu_{t-1})\in\Fqm^t$, we have
\begin{align*}
 &B_{\mu}(x,y)+B_{\mu}^{\pr}(x,y) \\
=&\Tqmq(\sum_{j=l}^{\frac{m+t-1}{2}}\mu_{j-l}(x^{q^{sj}}y+xy^{q^{sj}}))\\
=&\Tqmq(\sum_{j=\hmpo}^{\frac{m+t-1}{2}}\mu_{j-l}(x^{q^{sj}}y+xy^{q^{sj}}))+\Tqmq(\sum_{j=l}^{\frac{m-1}{2}}\mu_{j-l}(x^{q^{sj}}y+xy^{q^{sj}}))\\
=&\Tqmq(\sum_{j=\hmpo}^{\frac{m+t-1}{2}}\mu_{j-l}(x^{q^{sj}}y+xy^{q^{sj}}))+\Tqmq(\sum_{j=\hmpo}^{\frac{m+t-1}{2}}\mu_{m-j-l}(x^{q^{s(m-j)}}y+xy^{q^{s(m-j)}}))\\
=&\Tqmq(\sum_{j=\hmpo}^{\frac{m+t-1}{2}}\mu_{j-l}(x^{q^{sj}}y+xy^{q^{sj}}))+\Tqmq(\sum_{j=\hmpo}^{\frac{m+t-1}{2}}\mu_{m-j-l}^{q^{sj}}(x^{q^{sj}}y+xy^{q^{sj}}))\\
=&\Tqmq(\sum_{j=\hmpo}^{\frac{m+t-1}{2}}(\mu_{j-l}+\mu_{m-j-l}^{q^{sj}})(x^{q^{sj}}y+xy^{q^{sj}})).
\end{align*}
Thus,
\begin{align*}
&\{D+D^{\pr} \mid D \in \{B_{\mu} \mid \mu \in \Fqm^t\}\} \\ =&\{\Tqmq(\sum_{j=\hmpo}^{\frac{m+t-1}{2}}(\mu_{j-l}+\mu_{m-j-l}^{q^{sj}})(x^{q^{sj}}y+xy^{q^{sj}})) \mid \mu \in \Fqm^t\}.
\end{align*}
Note that $\{\mu_{j-l}, \mu_{m-j-l} \mid \hmpo \le j \le \frac{m+t-1}{2}\}=\{\mu_j \mid 0 \le j \le t-1\}$. When $\mu_{j-l}$ and $\mu_{m-j-l}$ range over $\Fqm$, $\mu_{j-l}+\mu_{m-j-l}^{q^{sj}}$ ranges over $\Fqm$ for $q^m$ times. Then each element of $\cY_1$ occurs a constant number of times in $\{D+D^{\pr} \mid D \in \{B_{\mu} \mid \mu \in \Fqm^t\}\}$. Hence, the number of elements in $\cY_1$ that are an extension of $A$, is a constant independent of $U$ and $A$. By Proposition~\ref{prop-tdesign}, we have shown that $\cY_1$ is a $t$-design in $\Sym(m,q)$, which completes the proof of 1).

Let $m$ be even. Given $\mu=(\mu_0,\mu_1,\cdots,\mu_{t-1})\in\Fqm^t$, we have
\begin{align*}
 &B_{\mu}(x,y)+B_{\mu}^{\pr}(x,y) \\
=&\Tqmq(\sum_{j=l}^{\frac{m+t-1}{2}}\mu_{j-l}(x^{q^{sj}}y+xy^{q^{sj}}))\\
=&\Tqmq(\mu_{\hm-l}(x^{q^{\frac{sm}{2}}}y+xy^{q^{\frac{sm}{2}}}))+\Tqmq(\sum_{j=\hmpt}^{\frac{m+t-1}{2}}\mu_{j-l}(x^{q^{sj}}y+xy^{q^{sj}}))\\
&+\Tqmq(\sum_{j=l}^{\hmmt}\mu_{j-l}(x^{q^{sj}}y+xy^{q^{sj}}))\\
=&\Tqhmq((\mu_{\hm-l}+\mu_{\hm-l}^{q^{\hm}})(x^{q^{\frac{sm}{2}}}y+xy^{q^{\frac{sm}{2}}}))+\Tqmq(\sum_{j=\hmpt}^{\frac{m+t-1}{2}}\mu_{j-l}(x^{q^{sj}}y+xy^{q^{sj}}))\\ &+\Tqmq(\sum_{j=\hmpt}^{\frac{m+t-1}{2}}\mu_{m-j-l}(x^{q^{s(m-j)}}y+xy^{q^{s(m-j)}}))\\
=&\Tqmq((\mu_{\hm-l}+\mu_{\hm-l}^{q^{\hm}})x^{q^{\frac{sm}{2}}}y)+\Tqmq(\sum_{j=\hmpt}^{\frac{m+t-1}{2}}\mu_{j-l}(x^{q^{sj}}y+xy^{q^{sj}}))\\
&+\Tqmq(\sum_{j=\hmpt}^{\frac{m+t-1}{2}}\mu_{m-j-l}^{q^{sj}}(x^{q^{sj}}y+xy^{q^{sj}}))\\
=&\Tqmq((\mu_{\hm-l}+\mu_{\hm-l}^{q^{\hm}})x^{q^{\frac{sm}{2}}}y+\sum_{j=\hmpt}^{\frac{m+t-1}{2}}(\mu_{j-l}+\mu_{m-j-l}^{q^{sj}})(x^{q^{sj}}y+xy^{q^{sj}})).
\end{align*}
Thus,
\begin{align*}
&\{D+D^{\pr} \mid D \in \{B_{\mu} \mid \mu \in \Fqm^t\}\} \\  =&\{\Tqmq((\mu_{\hm-l}+\mu_{\hm-l}^{q^{\hm}})x^{q^{\frac{sm}{2}}}y+\sum_{j=\hmpt}^{\frac{m+t-1}{2}}(\mu_{j-l}+\mu_{m-j-l}^{q^{sj}})(x^{q^{sj}}y+xy^{q^{sj}})) \mid \mu \in \Fqm^t\}.
\end{align*}
Note that $\{\mu_{\hm-l}, \mu_{j-l}, \mu_{m-j-l} \mid \hmpt \le j \le \frac{m+t-1}{2}\}=\{\mu_j \mid 0 \le j \le t-1\}$. When $\mu_{\hm-l}$ ranges over $\Fqm$, $\mu_{\hm-l}+\mu_{\hm-l}^{q^{\hm}}$ ranges over $\Fqhm$ for $q^{\hm}$ times.  Meanwhile, when $\mu_{j-l}$ and $\mu_{m-j-l}$ range over $\Fqm$, $\mu_{j-l}+\mu_{m-j-l}^{q^{sj}}$ ranges over $\Fqm$ for $q^m$ times. Then each element of $\cY_2$ occurs a constant number of times in $\{D+D^{\pr} \mid D \in \{B_{\mu} \mid \mu \in \Fqm^t\}\}$. Hence, the number of elements in $\cY_2$ that are an extension of $A$, is a constant independent of $U$ and $A$. By Proposition~\ref{prop-tdesign}, we have shown that $\cY_2$ is a $t$-design in $\Sym(m,q)$, which completes the proof of 2).
\end{proof}

Now, we are ready to derive the inner distributions of $\cS_1$ and $\cS_2$.

\begin{proposition}\label{prop-S}
For $j \in \{1,2\}$, let $(a_0^j,a_{1,1}^j,a_{1,-1}^j,a_{2,1}^j,a_{2,-1}^j,\ldots,a_{m,1}^j,a_{m,-1}^j)$ be the inner distribution of $\cS_j$, defined in (\ref{eqn-S1}) and (\ref{eqn-S2}).
\begin{itemize}
\item[1)] For $m$ being odd and $\hmmt \le i \le m-2$, $\cS_1$ is a proper $(2m-2i-1)$-code and $(2i+3-m)$-design in $\Sym(m,q)$. The inner distribution of $\cS_1$ satisfies (\ref{eqn-innerdis1}), in which $|Y|=|\cS_1|=q^{m(i-\frac{m-3}{2})}$. In particular, $a_{2m-2i-1,\tau}^1>0$ for $\tau \in \{1,-1\}$.
\item[2)] For $m$ being even and $\hmmt \le i \le m-2$, $\cS_2$ is a proper $(2m-2i-2)$-code and $(2i+3-m)$-design in $\Sym(m,q)$. The inner distribution of $\cS_2$ satisfies (\ref{eqn-innerdis2}), in which $|Y|=|\cS_2|=q^{m(i-\frac{m-3}{2})}$. In particular, for $\tau \in \{1,-1\}$ satisfying $\tau\eta(-1)^{m-i-1}=1$, we have $a_{2m-2i-2,\tau}^2=0$ and $a_{2m-2i-2,-\tau}^2>0$.
\end{itemize}
\end{proposition}
\begin{proof}
1) Let $\la_j \in \Fqm$, where $\hmpo \le j \le i+1$. Denote
$$
B_{\la}(x,y)=\Tqmq\Big(\sum_{j=\hmpo}^{i+1}(\la_jx^{q^j}+\la_j^{q^{-j}}x^{q^{-j}})y\Big) \in \cS_1.
$$
Note that
\begin{align*}
\Rad B_{\la}&=\{x \in \Fqm \mid \Tqmq\Big(\sum_{j=\hmpo}^{i+1}(\la_jx^{q^j}+\la_j^{q^{-j}}x^{q^{-j}})y\Big)=0, \forall y \in \Fqm\} \\
      &=\{x \in \Fqm \mid \sum_{j=\hmpo}^{i+1}(\la_jx^{q^j}+\la_j^{q^{m-j}}x^{q^{m-j}})=0\} \\
      &=\{x \in \Fqm \mid \sum_{j=\hmpo}^{i+1}(\la_j^{q^{-(m-i-1)}}x^{q^{j+i+1-m}}+\la_j^{q^{i+1-j}}x^{q^{i+1-j}})=0\}\\
      &=\{x \in \Fqm \mid \sum_{j=\hmpo}^{i+1}(\la_j^{q^{i+1}}x^{q^{j+i+1-m}}+\la_j^{q^{i+1-j}}x^{q^{i+1-j}})=0\}.
\end{align*}
Since $\sum_{j=\hmpo}^{i+1}(\la_j^{q^{i+1}}x^{q^{j+i+1-m}}+\la_j^{q^{i+1-j}}x^{q^{i+1-j}})$ is a linearized polynomial over $\Fqm$ with degree at most $q^{2i+2-m}$. Then $\Rad B_{\la}$ forms a vector space over $\Fq$ with dimension at most $2i+2-m$ \cite[Theorem 3.50]{LN}. Hence, for each $B_{\la} \in \cS_1$, we have shown $\rank(B_{\la})=m-\dim\Rad B_{\la} \ge 2m-2i-2$. Thus, $\cS_1$ is a $(2m-2i-2)$-code in $\Sym(m,q)$.

Applying 1) of Proposition~\ref{prop-tdesignfamily} with $t=2i+3-m$ and $s=1$, we know that $\cS_1$ is a $(2i+3-m)$-design in $\Sym(m,q)$. Thus, $\cS_1$ is a $(2i+2-m,\eta(-1)^{i-\frac{m-3}{2}})$-design by definition. Employing (\ref{eqn-innerdis3}) in Proposition~\ref{prop-innereven} with $l=m-i-1$, together with the fact that $|\cS_1|=q^{m(i-\frac{m-3}{2})}$, we can see that $a_{2m-2i-2,\tau}^1=0$ for $\tau \in \{1,-1\}$. Thus, $\cS_1$ is a $(2m-2i-1)$-code. Employing (\ref{eqn-innerdis1}) in Proposition~\ref{prop-innerodd} with $l=m-i$, we have $a_{2m-2i-1,\tau}^1>0$ for $\tau \in \{1,-1\}$. Therefore, $\cS_1$ is a proper $(2m-2i-1)$-code.

2) Using similar arguments as in the proof of 1), we can show that $\cS_2$ is a $(2m-2i-2)$-code. Applying 2) of Proposition~\ref{prop-tdesignfamily} with $t=2i+3-m$ and $s=1$, we know that $\cS_2$ is a $(2i+3-m)$-design in $\Sym(m,q)$. Employing (\ref{eqn-innerdis2}) of Proposition~\ref{prop-innereven}, we have $a_{2m-2i-2,\tau}^2=0$ and $a_{2m-2i-2,-\tau}^2>0$, where $\tau\eta(-1)^{m-i-1}=1$. Thus, $\cS_2$ is a proper $(2m-2i-2)$-code.
\end{proof}

Now we are ready to prove the following main result of this section.

\begin{theorem}\label{thm-mainodd}
Let $q$ be an odd prime power and $m$ be a positive integer satisfying
$$
\begin{cases}
  m \ge 2 & \mbox{if $q=3$}, \\
  m \ge 1 & \mbox{if $q \ge 5$}.
\end{cases}
$$
\begin{itemize}
\item[1)] For $m$ being odd and $i$ being nonnegative integer with $\frac{m-2}{2} \le i \le m-\ltm-1$, the BCH code $\codespe$ has minimum distance $\de_i=q^m-q^{m-1}-q^i-1$. Moreover, its weight enumerator is
    $$
    T(Z)+\sum_{r=2m-2i-1}^m\sum_{\tau \in \{1,-1\}} a_{r,\tau}^1 U_{r,\tau}(Z),
    $$
    where $a_{r,\tau}^1$ is the same as $a_{r,\tau}$ defined in (\ref{eqn-innerdis1}) with $|Y|=q^{m(i-\frac{m-3}{2})}$, $T(Z)$ is defined in (\ref{eqn-weightenum}) and $U_{r,\tau}(Z)$ is defined in Proposition~\ref{prop-Qweightodd}.
\item[2)] For $m$ being even and $i$ being nonnegative integer with $\frac{m-2}{2} \le i \le m-\ltm-1$, the BCH code $\codespe$ has minimum distance $\de_i=q^m-q^{m-1}-q^i-1$. Moreover, its weight enumerator is
    $$
    T(Z)+\sum_{r=2m-2i-2}^m\sum_{\tau \in \{1,-1\}} a_{r,\tau}^2 U_{r,\tau}(Z),
    $$
    where $a_{r,\tau}^2$ is the same as $a_{r,\tau}$ defined in (\ref{eqn-innerdis2}) with $|Y|=q^{m(i-\frac{m-3}{2})}$, $T(Z)$ is defined in (\ref{eqn-weightenum}) and $U_{r,\tau}(Z)$ is defined in Proposition~\ref{prop-Qweightodd}.
\end{itemize}
\end{theorem}
\begin{proof}
When $q=3$, we need $m \ge 2$ to ensure that $\codespe$ is well-defined.

1) By 1) of Proposition~\ref{prop-innerdis}, $\cQ_1$ and $\cS_1$ have the same inner distribution $(a_{r,\tau}^1)$, which has been obtained in 1) of Proposition~\ref{prop-S}. Since $\cQ_1$ is an additive subset, $\cQ_1$ contains exactly one zero quadratic form and $a_{r,\tau}^1$ is the number of quadratic forms in $\cQ_1$ with rank $r$ and type $\tau$. Together with Propositions~\ref{prop-cosetunion} and \ref{prop-Qweightodd}, we can see the weight enumerator is
$$
T(Z)+\sum_{r=2m-2i-1}^m\sum_{\tau \in \{1,-1\}} a_{r,\tau}^1U_{r,\tau}(Z),
$$
where $a_{r,\tau}^1$ is the same as $a_{r,\tau}$ defined in (\ref{eqn-innerdis1}) with $|Y|=|\cS_1|=q^{m(i-\frac{m-3}{2})}$. Let $\tau \in \{1 ,-1\}$ satisfy $\tau\eta(-1)^{m-i-1}=1$. By 1) of Proposition~\ref{prop-S}, we have $a_{2m-2i-1,\tau}^1>0$. Namely, there is a quadratic form $Q \in \cQ_1$ with rank $2m-2i-1$ and type $\tau$. By Proposition~\ref{prop-Qweightodd}, the subcode $(Q(x))_{x\in\Fqm^*}+\PRM_q(1,m)$ contains a codeword of weight $\de_i=q^m-q^{m-1}-q^i-1$. Note that by Proposition~\ref{prop-traceprop}, $\codespe$ has Bose distance $\de_i$. Hence, the minimum distance of $\codespe$ is equal to $\de_i=q^m-q^{m-1}-q^i-1$.

2) By 1) of Proposition~\ref{prop-innerdis}, $\cQ_2$ and $\cS_2$ have the same inner distribution $(a_{r,\tau}^2)$, which has been obtained in 2) of Proposition~\ref{prop-S}. Since $\cQ_2$ is an additive subset, $\cQ_2$ contains exactly one zero quadratic form and $a_{r,\tau}^2$ is the number of quadratic forms in $\cQ_2$ with rank $r$ and type $\tau$. Together with Propositions~\ref{prop-cosetunion} and \ref{prop-Qweightodd}, we can see the weight enumerator is
$$
T(Z)+\sum_{r=2m-2i-2}^m\sum_{\tau \in \{1,-1\}} a_{r,\tau}^2U_{r,\tau}(Z),
$$
where $a_{r,\tau}^2$ is the same as $a_{r,\tau}$ defined in (\ref{eqn-innerdis2}) with $|Y|=|\cS_2|=q^{m(i-\frac{m-3}{2})}$. Let $\tau \in \{1 ,-1\}$ satisfy $\tau\eta(-1)^{m-i-1}=1$. By 2) of Proposition~\ref{prop-S}, we have $a_{2m-2i-2,-\tau}^2>0$. Namely, there is a quadratic form $Q \in \cQ_2$ with rank $2m-2i-2$ and type $-\tau$. By Proposition~\ref{prop-Qweightodd}, the subcode $(Q(x))_{x\in\Fqm^*}+\PRM_q(1,m)$ contains a codeword of weight $\de_i=q^m-q^{m-1}-q^i-1$. Note that by Proposition~\ref{prop-traceprop}, $\codespe$ has Bose distance $\de_i$. Hence, the minimum distance of $\codespe$ is equal to $\de_i=q^m-q^{m-1}-q^i-1$.
\end{proof}

\begin{remark}
Let $p$ be an odd prime. The weight distributions of the even-like subcodes of $\cC_{(p,m,\de_{\hmmo})}$ and $\cC_{(p,m,\de_{\hmpo})}$ with $m$ being odd, as well as the weight distributions of the even-like subcodes of $\cC_{(p,m,\de_{\hmmt})}$ and $\cC_{(p,m,\de_{\hm})}$ with $m$ being even, have been obtained in {\rm\cite{Ding3}}.
\end{remark}

\begin{example}
For $q=3$ and $m=3$, consider ternary narrow-sense primitive BCH code $\cC_{(3,3,\de_1)}$. Numerical experiment shows that $\cC_{(3,3,\de_1)}$ has weight enumerator
$$
1+390Z^{14}+312Z^{15}+520Z^{17}+260Z^{18}+546Z^{20}+156Z^{21}+2Z^{26},
$$
which is consistent with Theorem~\ref{thm-mainodd}. $\cC_{(3,3,\de_1)}$ is a ternary $[26,7,14]$ code. According to the codetable {\rm \cite{Gra}} (see also {\rm \cite[p. 300, Table A.92]{Ding15}}), it has the largest possible minimum distance among all ternary linear codes with length $26$ and dimension $7$.
\end{example}

\begin{example}
For $q=3$ and $m=4$, consider ternary narrow-sense primitive BCH code $\cC_{(3,4,\de_2)}$. Numerical experiment shows that $\cC_{(3,4,\de_2)}$ has weight enumerator
\begin{align*}
&1+3800Z^{44}+3040Z^{45}+14400Z^{47}+9900Z^{48}+17136Z^{50}+10080Z^{51}\\
&+33280Z^{53}+16640Z^{54}+34200Z^{56}+14400Z^{57}+10080Z^{59}+3528Z^{60}\\
&+5040Z^{62}+1440Z^{63}+160Z^{71}+20Z^{72}+2Z^{80},
\end{align*}
which is consistent with Theorem~\ref{thm-mainodd}. $\cC_{(3,4,\de_2)}$ is a ternary $[80,11,44]$ code. According to the codetable {\rm \cite{Gra}}, it has the largest minimum distance among all known ternary linear codes with length $80$ and dimension $11$.
\end{example}

\section{Determining the minimum distance when $q$ is even}\label{sec5}

In this section, we are going to show that the minimum distance of $\codespe$ is equal to its Bose distance $\de_i$, when $q$ is an even prime power. Throughout the rest of this section, we always assume that $q$ is even.

By Proposition~\ref{prop-cosetunion}, in order to determine the minimum distance of $\codespe$, it suffices to compute the minimum weight of each subcode $(Q(x))_{x \in \Fqm^*}+\PRM_q(1,m)$ for each $Q \in \cQ_1$ when $m$ is odd and for each $Q \in \cQ_2$ when $m$ is even. As the first step, the following proposition shows when $q$ is even, the weight distribution of the subcode $(Q(x))_{x \in \Fqm^*}+\PRM_q(1,m)$ depends only on the rank and the type of the quadratic form $Q$ on $\Fqm$.

\begin{proposition}\label{prop-Qweighteven}
Let $q$ be even and $Q$ be a quadratic form on $\Fqm$, whose rank is at least $1$.
\begin{itemize}
\item[1)] If $Q$ has rank $2r$ and type $0$, then the weight enumerator of $(Q(x))_{x \in \Fqm^*}+\PRM_q(1,m)$ is
\begin{align*}
W_{2r,0}(Z)=&(q-1)(q^{2r-1}-q^{r-1})Z^{q^m-q^{m-1}-q^{m-r-1}(q-1)-1}\\
&+(q^{2r-1}+q^{r-1}(q-1))Z^{q^m-q^{m-1}-q^{m-r-1}(q-1)}\\
&+(q-1)(q^m-q^{2r})Z^{q^m-q^{m-1}-1}+(q^m-q^{2r})Z^{q^m-q^{m-1}}\\
&+(q-1)((q-1)q^{2r-1}+q^{r-1})Z^{q^m-q^{m-1}+q^{m-r-1}-1}\\
&+(q-1)(q^{2r-1}-q^{r-1})Z^{q^m-q^{m-1}+q^{m-r-1}}.
\end{align*}
\item[2)] If $Q$ has rank $2r+1$ and type $1$, then the weight enumerator of $(Q(x))_{x \in \Fqm^*}+\PRM_q(1,m)$ is
\begin{align*}
W_{2r+1,1}(Z)=&(q-1)\frac{q^{2r+1}-q^{2r}-q^r}{2}Z^{q^m-q^{m-1}-q^{m-r-1}-1}\\
&+(q-1)\frac{q^{2r}+q^r}{2}Z^{q^m-q^{m-1}-q^{m-r-1}}\\
&+(q-1)(q^m-q^{2r+1}+q^{2r})Z^{q^m-q^{m-1}-1}\\
&+(q^m-q^{2r+1}+q^{2r})Z^{q^m-q^{m-1}}\\
&+(q-1)\frac{q^{2r+1}-q^{2r}+q^r}{2}Z^{q^m-q^{m-1}+q^{m-r-1}-1}\\
&+(q-1)\frac{q^{2r}-q^r}{2}Z^{q^m-q^{m-1}+q^{m-r-1}}.
\end{align*}
\item[3)] If $Q$ has rank $2r$ and type $2$, then the weight enumerator of $(Q(x))_{x \in \Fqm^*}+\PRM_q(1,m)$ is
\begin{align*}
W_{2r,2}(Z)=&(q-1)((q-1)q^{2r-1}-q^{r-1})Z^{q^m-q^{m-1}-q^{m-r-1}-1}\\
&+(q-1)(q^{2r-1}+q^{r-1})Z^{q^m-q^{m-1}-q^{m-r-1}}\\
&+(q-1)(q^m-q^{2r})Z^{q^m-q^{m-1}-1}+(q^m-q^{2r})Z^{q^m-q^{m-1}}\\
&+(q-1)(q^{2r-1}+q^{r-1})Z^{q^m-q^{m-1}+q^{m-r-1}(q-1)-1}\\
&+(q^{2r-1}-q^{r-1}(q-1))Z^{q^m-q^{m-1}+q^{m-r-1}(q-1)}.
\end{align*}
\end{itemize}
The above weight enumerators are also summarized in Tables \ref{tab-qevent0}, \ref{tab-qevent1} and \ref{tab-qevent2}.
\end{proposition}
\begin{proof}
Let $\cL$ be the set of all homogenous linear functions on $\Fqm$. Then, the weight enumerator of $(Q(x))_{x \in \Fqm^*}+\PRM_q(1,m)$ can be read from the multiset $\{\wt(Q+L+c) \mid L \in \cL, c \in \Fq\}$. The conclusion follows from (\ref{eqn-fweight}) and Lemma~\ref{lem-quadevenset}.
\end{proof}

\begin{table}
\begin{center}
\ra{1.5}\caption{Weight distribution of $(Q(x))_{x \in \Fqm^*}+\PRM_q(1,m)$, $q$ even, $Q$ has rank $2r$ and type $0$}
\begin{tabular}{|c|c|}
\hline
Weight  &   Frequency  \\ \hline
$q^m-q^{m-1}-q^{m-r-1}(q-1)-1$ & $(q-1)(q^{2r-1}-q^{r-1})$ \\ \hline
$q^m-q^{m-1}-q^{m-r-1}(q-1)$ & $q^{2r-1}+q^{r-1}(q-1)$ \\ \hline
$q^m-q^{m-1}-1$ & $(q-1)(q^m-q^{2r})$ \\ \hline
$q^m-q^{m-1}$ & $q^m-q^{2r}$ \\ \hline
$q^m-q^{m-1}+q^{m-r-1}-1$ & $(q-1)((q-1)q^{2r-1}+q^{r-1})$ \\ \hline
$q^m-q^{m-1}+q^{m-r-1}$ & $(q-1)(q^{2r-1}-q^{r-1})$ \\ \hline
\end{tabular}
\end{center}
\label{tab-qevent0}
\end{table}

\begin{table}
\begin{center}
\ra{1.5}\caption{Weight distribution of $(Q(x))_{x \in \Fqm^*}+\PRM_q(1,m)$, $q$ even, $Q$ has rank $2r+1$ and type $1$}
\begin{tabular}{|c|c|}
\hline
Weight  &   Frequency  \\ \hline
$q^m-q^{m-1}\pm q^{m-r-1}-1$ & $(q-1)\frac{q^{2r+1}-q^{2r}\pm q^r}{2}$ \\ \hline
$q^m-q^{m-1}\pm q^{m-r-1}$ & $(q-1)\frac{q^{2r}\mp q^r}{2}$ \\ \hline
$q^m-q^{m-1}-1$ & $(q-1)(q^m-q^{2r+1}+q^{2r})$ \\ \hline
$q^m-q^{m-1}$ & $q^m-q^{2r+1}+q^{2r}$ \\ \hline
\end{tabular}
\end{center}
\label{tab-qevent1}
\end{table}

\begin{table}
\begin{center}
\ra{1.5}\caption{Weight distribution of $(Q(x))_{x \in \Fqm^*}+\PRM_q(1,m)$, $q$ even, $Q$ has rank $2r$ and type $2$}
\begin{tabular}{|c|c|}
\hline
Weight  &   Frequency  \\ \hline
$q^m-q^{m-1}-q^{m-r-1}-1$ & $(q-1)((q-1)q^{2r-1}-q^{r-1})$ \\ \hline
$q^m-q^{m-1}-q^{m-r-1}$ & $(q-1)(q^{2r-1}+q^{r-1})$ \\ \hline
$q^m-q^{m-1}-1$ & $(q-1)(q^m-q^{2r})$ \\ \hline
$q^m-q^{m-1}$ & $q^m-q^{2r}$ \\ \hline
$q^m-q^{m-1}+q^{m-r-1}(q-1)-1$ & $(q-1)(q^{2r-1}+q^{r-1})$ \\ \hline
$q^m-q^{m-1}+q^{m-r-1}(q-1)$ & $q^{2r-1}-q^{r-1}(q-1)$ \\ \hline
\end{tabular}
\end{center}
\label{tab-qevent2}
\end{table}

For $\cQ_1$ and $\cQ_2$, define the corresponding sets of alternating bilinear forms as
$$
\cA_j=\{B_Q \mid Q \in \cQ_j\},
$$
where $B_Q(x,y)=Q(x+y)-Q(x)-Q(y)$ and $j=1,2$. More precisely, we have
$$
  \cA_1:=\cA_1(i)=\Bigg\{\Tqmq\Big(\sum_{j=\hmpo}^{i+1}(\la_jx^{q^j}+\la_j^{q^{-j}}x^{q^{-j}})y\Big) \mid \la_{\hmpo},\ldots,\la_{i+1} \in \Fqm \Bigg\}
$$
and
\begin{align*}
  \cA_2:=\cA_2(i)=\Bigg\{\Tqmq\Big(\big(\la_{\hm}x^{q^{\hm}}+\sum_{j=\hmpt}^{i+1}&(\la_jx^{q^j}+\la_j^{q^{-j}}x^{q^{-j}})\big)y\Big) \\
   &\Big | \la_{\hm} \in \Fqhm, \la_{\hmpt},\ldots,\la_{i+1} \in \Fqm \Bigg\}
\end{align*}
with $|\cA_1|=|\cA_2|=q^{m(i-\frac{m-3}{2})}$.

\begin{proposition}\label{prop-A}
For $j \in \{1,2\}$, let $(b_0^j,b_2^j,b_4^j,\ldots,b_{2\lfloor \frac{m}{2} \rfloor}^j)$ be the inner distribution of $\cA_j$.
\begin{itemize}
\item[1)] For $m$ being odd and $\hmmt \le i \le m-2$, $\cA_1$ is a proper $(2m-2i-2)$-code in $\Alt(m,q)$. Thus, $b_{2m-2i-2}^1>0$.
\item[2)] For $m$ being even and $\hmmt \le i \le m-2$, $\cA_2$ is a proper $(2m-2i-2)$-code in $\Alt(m,q)$. Thus, $b_{2m-2i-2}^2>0$.
\end{itemize}
\end{proposition}
\begin{proof}
1) Using similar arguments as in the proof of Proposition~\ref{prop-S}, we can show $\cA_1$ is a $(2m-2i-2)$-code in $\Alt(m,q)$. Assume $b_{2m-2i-2}^1=0$, then $\cA_1$ is a $(2m-2i)$-code in $\Alt(m,q)$. By Proposition~\ref{prop-bound}, we must have $|\cA_1| \le q^{m(i-\frac{m-1}{2})}$. This contradicts to $|\cA_1|=q^{m(i-\frac{m-3}{2})}$. Thus, $b_{2m-2i-2}^1>0$ and $\cA_1$ is a proper $(2m-2i-2)$-code.

2) Using similar arguments as in the proof of Proposition~\ref{prop-S}, we can show $\cA_2$ is a $(2m-2i-2)$-code in $\Alt(m,q)$. Assume $b_{2m-2i-2}^2=0$, then $\cA_2$ is a $(2m-2i)$-code in $\Alt(m,q)$. By Proposition~\ref{prop-bound}, we must have $|\cA_2| \le q^{(m-1)(i-\frac{m-2}{2})}$. This contradicts to $|\cA_2|=q^{m(i-\frac{m-3}{2})}$. Thus, $b_{2m-2i-2}^2>0$ and $\cA_2$ is a proper $(2m-2i-2)$ code.
\end{proof}

Now we are ready to prove the main result of this section.

\begin{theorem}\label{thm-maineven}
Let $q$ be an even prime power and $m$ be a positive integer satisfying
$$
\begin{cases}
  m \ge 3 & \mbox{if $q=2$}, \\
  m \ge 1 & \mbox{if $q\ge 4$}.
\end{cases}
$$
For nonnegative integer $i$ with $\frac{m-2}{2} \le i \le m-\ltm-1$, the BCH code $\codespe$ has minimum distance $\de_i=q^m-q^{m-1}-q^i-1$.
\end{theorem}
\begin{proof}
When $q=2$, we need $m \ge 3$ to ensure that $\codespe$ is well-defined. We only prove the theorem for the case with $m$ being odd, since the case with $m$ being even is analogous.

When $m$ is odd, by Proposition~\ref{prop-cosetunion}, the minimum distance of $\codespe$ is equal to the minimum weight of codewords belonging to $\cup_{Q \in \cQ_1}((Q(x))_{\Fqm^*}+\PRM_{q}(1,m))$.   For $0 \le i \le \lhm$, $\tau \in \{0,1,2\}$ and $0 \le 2i+\tau \le m$, let $(d_{2i+\tau,\tau}^1)$ and $(b_{2i}^1)$ be the inner distributions of $\cQ_1$ and $\cA_1$ respectively. Since $\cQ_1$ is an additive subset, $d_{2i+\tau,\tau}^1$ is the number of quadratic forms in $\cQ_1$ with rank $2i+\tau$ and type $\tau$. By 1) of Proposition~\ref{prop-A}, we have $b_{2m-2i-2}^1>0$. By 2) of Proposition~\ref{prop-innerdis}, we have $d_{2m-2i-2,0}^1+d_{2m-2i-1,1}^1+d_{2m-2i-2,2}^1=b_{2m-2i-2}^1>0$. If $d_{2m-2i-2,0}^1>0$, then we have a quadratic form $Q_0 \in \cQ_1$ with rank $2m-2i-2$ and type $0$. By Proposition~\ref{prop-Qweighteven}, $(Q_0(x))_{\Fqm^*}+\PRM_q(1,m)$ contains a codeword with weight $q^m-q^{m-1}-q^{i}(q-1)-1$. This is impossible since the minimum distance of $\codespe$ is greater than or equal to the Bose distance $\de_i=q^m-q^{m-1}-q^i-1$. Thus, we have $d_{2m-2i-2,0}^1=0$. Consequently, $d_{2m-2i-1,1}^1+d_{2m-2i-2,2}^1=b_{2m-2i-2}^1>0$. Then, we have either $d_{2m-2i-1,1}^1>0$ or $d_{2m-2i-2,2}^1>0$. If $d_{2m-2i-1,1}^1>0$, we have a quadratic form $Q_1 \in \cQ_1$ with rank $2m-2i-1$ and type $1$. By Proposition~\ref{prop-Qweighteven}, $(Q_1(x))_{x\in\Fqm^*}+\PRM_{q}(1,m)$ contains a codeword with weight $q^m-q^{m-1}-q^i-1$. If $d_{2m-2i-2,2}^1>0$, we have a quadratic form $Q_2 \in \cQ_1$ with rank $2m-2i-2$ and type $2$. By Proposition~\ref{prop-Qweighteven}, $(Q_2(x))_{x\in\Fqm^*}+\PRM_{q}(1,m)$ contains a codeword with weight $q^m-q^{m-1}-q^i-1$. Hence, $\codespe$ must contain a codeword of weight $q^m-q^{m-1}-q^i-1$, which is equal to its Bose distance by Proposition~\ref{prop-traceprop}. Therefore, $\codespe$ has minimum distance $\de_i=q^m-q^{m-1}-q^i-1$.
\end{proof}

\begin{remark}
With delicate techniques employed, the weight distributions of following binary codes have been determined by Kasami {\rm \cite{K1}}:
\begin{itemize}
\item for $m \ge 4$ being even, $\cC_{(2,m,\de_{\hmmt})}$ and $\cC_{(2,m,\de_{\hm})}$
\item for $m \ge 3$ being odd, $\cC_{(2,m,\de_{\hmmo})}$
\item for $m \ge 5$ being odd, $\cC_{(2,m,\de_{\hmpo})}$
\item for $m \ge 11$ being odd, $\cC_{(2,m,\de_{\frac{m+3}{2}})}$
\end{itemize}

Morevoer, the weight distributions of the even-like subcodes of $\cC_{(2,m,\de_{\hmmo})}$ and $\cC_{(2,m,\de_{\hmpo})}$ with $m$ being odd, as well as the weight distributions of the even-like subcodes of $\cC_{(2,m,\de_{\hmmt})}$ and $\cC_{(2,m,\de_{\hm})}$ with $m$ being even, have been obtained in {\rm\cite{Ding3}}.
\end{remark}

\begin{example}
For $q=2$ and $m=6$, consider binary narrow-sense primitive BCH codes $\cC_{(2,6,\de_2)}$ and $\cC_{(2,6,\de_3)}$. Numerical experiments show that $\cC_{(2,6,\de_2)}$ has minimum distance $\de_2=27$ and $\cC_{(2,6,\de_3)}$ has minimum distance $\de_3=23$, which are consistent with Theorem~\ref{thm-maineven}. Moreover, according to the codetable {\rm \cite{Gra}} (see also {\rm \cite[p. 258, Table A.31]{Ding15}}), $\cC_{(2,6,\de_2)}$ has the largest minimum distance among all known binary linear codes with length $63$ and dimension $10$. $\cC_{(2,6,\de_3)}$ has the largest minimum distance among all known binary linear codes with length $63$ and dimension $16$.
\end{example}

Combining Theorem~\ref{thm-mainodd} and Theorem~\ref{thm-maineven}, we obtain our main result Theorem~\ref{thm-main}.

\section{Concluding remarks}\label{sec6}


The determination of the minimum distance of BCH codes is a very challenging problem. Even for the most well-studied narrow-sense primitive BCH codes, despite some experimental results, very few theoretical results on the minimum distance have been found in the last four decades. In this paper, we make progress on this problem by determining the minimum distance of $q$-ary narrow-sense BCH codes with length $q^m-1$ and Bose distance $q^m-q^{m-1}-q^i-1$, where $\hmmt \le i \le m-\ltm-1$. Note that when $q=2$, this result has been obtained in the classical works of Berlekamp \cite{Ber70} and Kasami \cite{K2}, in which many delicate techniques are involved. We generalize their results in a neat and unified way, based on the well-rounded theories of sets of quadratic forms, symmetric bilinear forms and alternating bilinear forms over finite fields. Our result depends heavily on the theory of sets of symmetric bilinear forms over finite fields with odd characteristic \cite{Sch15} and the theory of sets of alternating bilinear forms over finite fields \cite{DG}. For the theory of sets of symmetric bilinear forms over finite fields with even characteristic and the theory of sets of unrestricted bilinear forms over finite fields, please refer to \cite{Sch10} and \cite{Del78} respectively.

Recall that when $q$ is odd, the weight distribution of $\codespe$ has been obtained. The key point is that the inner distributions of $\cS_1$ and $\cS_2$ determine those of $\cQ_1$ and $\cQ_2$, by 1) of Proposition~\ref{prop-innerdis}. Thus, it is natural to ask if the weight distribution can also be obtained when $q$ is even. We remark that the inner distribution of $\cA_1$ and $\cA_2$ can be computed using the same method displayed in \cite{Sch10}. However, by 2) of Proposition~\ref{prop-innerdis}, the inner distributions of $\cA_1$ and $\cA_2$ do not determine those of $\cQ_1$ and $\cQ_2$. Thus, in order to compute the weight distribution when $q$ is even, more precise information of the inner distributions of $\cQ_1$ and $\cQ_2$ is required. A possible direction is to establish some general results on the inner distribution of certain subsets in $\Qua(n,q)$ when $q$ is even, from which the weight distribution of $\codespe$ follows. To the best of our knowledge, not much is known about the scheme $\Qua(n,q)$ \cite{FWMM,WWMM}. Hence, we leave this as an interesting future problem.

\appendix
\section{Number of zeroes to the summation of a quadratic form and a linear function}

Let $Q$ be a quadratic form on $\Fqm$ and $L$ be a homogenous linear function on $\Fqm$. In this appendix, we study the number of zeroes in $\Fqm$, to the equation $Q(x)+L(x)+c=0$, where $c \in \Fq$. These results are employed to complete the proof of Proposition~\ref{prop-Qweightodd} and Proposition~\ref{prop-Qweighteven}. The approach used here is similar to that of \cite{McE}, in which the author obtained the weight distribution of the second order Reed-Muller and the second order generalized Reed-Muller codes. We remark that the derivation in \cite{McE} is somewhat hasty, so that several mistakes were made in the paper. Therefore, we intend to present a detailed account below.

We first consider the case with $q$ being an odd prime power. We use $S$ to denote the set of all nonzero squares in $\Fq$ and $NS$ to denote the set of all nonsquares in $\Fq$.

The following lemma is a folklore. It can be derived, for instance, from the viewpoint of difference sets \cite[Chapter 7]{Jung} and partial difference sets \cite{Ma}.

\begin{lemma}\label{lem-intersize}
For $b \in \Fq$, let $W(b)$ be a subset of $\Fq \times \Fq$ defined by
$$
W(b)=\{(h+b,h) \mid h \in \Fq\}.
$$
For the following nine subsets partitioning $\Fq\times\Fq$
\begin{align*}
U_1&=\{0\} \times \{0\}, \quad &U_2&=\{0\} \times S, \quad & U_3&=\{0\} \times NS, \\
U_4&=S \times \{0\},   \quad   &U_5&=S \times S,  \quad &   U_6&=S \times NS, \\
U_7&=NS \times \{0\},  \quad   &U_8&=NS \times S, \quad &   U_9&=NS \times NS,
\end{align*}
define $u_i=|W(b) \cap U_i|$, $1 \le i \le 9$. The values of the sequence $(u_i)_{i=1}^9$ depend only on $b$ and the exact values are listed in Table~\ref{tab-u}.
\end{lemma}

\begin{table}
\begin{center}
\ra{1.5}\caption{}
\begin{tabular}{|c|c|}
\hline
$b$  &   $(u_1,u_2,u_3,u_4,u_5,u_6,u_7,u_8,u_9)$  \\ \hline
$b=0$ & $(1,0,0,0,\frac{q-1}{2},0,0,0,\frac{q-1}{2})$ \\ \hline
$b \in S, -b \in S$ & $(0,1,0,1,\frac{q-5}{4},\frac{q-1}{4},0,\frac{q-1}{4},\frac{q-1}{4})$ \\ \hline
$b \in S, -b \in NS$ & $(0,0,1,1,\frac{q-3}{4},\frac{q-3}{4},0,\frac{q+1}{4},\frac{q-3}{4})$ \\ \hline
$b \in NS, -b \in S$ & $(0,1,0,0,\frac{q-3}{4},\frac{q+1}{4},1,\frac{q-3}{4},\frac{q-3}{4})$ \\ \hline
$b \in NS, -b \in NS$ & $(0,0,1,0,\frac{q-1}{4},\frac{q-1}{4},1,\frac{q-1}{4},\frac{q-5}{4})$ \\ \hline
\end{tabular}
\end{center}
\label{tab-u}
\end{table}

Recall that for a function $f$ from $\Fqm$ to $\Fq$, we use $N(f)$ to denote the number of zeroes in $\Fqm$ to the equation $f(x)=0$. We have the following lemma.

\begin{lemma}\label{lem-quadodd}
Let $q$ be an odd prime power and $Q$ be a quadratic form of rank $r \ge 1$ and type $\tau$ on $\Fqm$. Let $\cL$ be the set of all homogenous linear functions on $\Fqm$ and $c \in \Fq$. Suppose $f$ ranges over $\{Q+L+c \mid L \in \cL\}$. Then the following holds.
\noindent
(1) Let $r$ be odd. If $c=0$, then
$$
N(f)=\left\{
\begin{aligned}
  &q^{m-1} & \mbox{$q^m-q^{r}+q^{r-1}$ times,} \\
  &q^{m-1}\pm\tau\eta(-1)^{\frac{r-1}{2}}q^{m-\frac{r+1}{2}} & \mbox{$\frac{(q-1)}{2}(q^{r-1}\pm\tau\eta(-1)^{\frac{r-1}{2}}q^{\frac{r-1}{2}})$ times.}
\end{aligned}\right.
$$
If $c \in S$, then
$$
N(f)=\left\{
\begin{aligned}
  &q^{m-1} & \mbox{$q^m-q^{r}+q^{r-1}+\tau\eta(-1)^{\frac{r-1}{2}}q^{\frac{r-1}{2}}$ times,} \\
  &q^{m-1}+\tau\eta(-1)^{\frac{r-1}{2}}q^{m-\frac{r+1}{2}} & \mbox{$\frac{q-1}{2}q^{r-1}-\tau\eta(-1)^{\frac{r-1}{2}}q^{\frac{r-1}{2}}$ times,} \\
  &q^{m-1}-\tau\eta(-1)^{\frac{r-1}{2}}q^{m-\frac{r+1}{2}} & \mbox{$\frac{q-1}{2}q^{r-1}$ times.}
\end{aligned}\right.
$$
If $c \in NS$, then
$$
N(f)=\left\{
\begin{aligned}
  &q^{m-1} & \mbox{$q^m-q^{r}+q^{r-1}-\tau\eta(-1)^{\frac{r-1}{2}}q^{\frac{r-1}{2}}$ times,} \\
  &q^{m-1}+\tau\eta(-1)^{\frac{r-1}{2}}q^{m-\frac{r+1}{2}} & \mbox{$\frac{q-1}{2}q^{r-1}$ times,} \\
  &q^{m-1}-\tau\eta(-1)^{\frac{r-1}{2}}q^{m-\frac{r+1}{2}} & \mbox{$\frac{q-1}{2}q^{r-1}+\tau\eta(-1)^{\frac{r-1}{2}}q^{\frac{r-1}{2}}$ times.}
\end{aligned}\right.
$$
\noindent
(2) Let $r$ be even. If $c=0$, then
$$
N(f)=\left\{
\begin{aligned}
  &q^{m-1} & \mbox{$q^m-q^{r}$ times,} \\
  &q^{m-1}+\tau\eta(-1)^{\frac{r}{2}}q^{m-\frac{r+2}{2}}(q-1) &\mbox{$q^{r-1}+\tau\eta(-1)^{\frac{r}{2}}q^{\frac{r-2}{2}}(q-1)$ times,}\\
  &q^{m-1}-\tau\eta(-1)^{\frac{r}{2}}q^{m-\frac{r+2}{2}} & \mbox{$(q-1)(q^{r-1}-\tau\eta(-1)^{\frac{r}{2}}q^{\frac{r-2}{2}})$ times.}
\end{aligned}\right.
$$
If $c\ne0$, then
$$
N(f)=\left\{
\begin{aligned}
  &q^{m-1} & \mbox{$q^m-q^{r}$ times,} \\
  &q^{m-1}+\tau\eta(-1)^{\frac{r}{2}}q^{m-\frac{r+2}{2}}(q-1) &\mbox{$q^{r-1}-\tau\eta(-1)^{\frac{r}{2}}q^{\frac{r-2}{2}}$ times,}\\
  &q^{m-1}-\tau\eta(-1)^{\frac{r}{2}}q^{m-\frac{r+2}{2}} & \mbox{$(q-1)q^{r-1}+\tau\eta(-1)^{\frac{r}{2}}q^{\frac{r-2}{2}}$ times.}
\end{aligned}\right.
$$
\end{lemma}
\begin{proof}
Since $Q$ has rank $r$ and type $\tau$, by 1) of Proposition~\ref{prop-equivalence}, $Q$ is equivalent to a canonical quadratic form $\sum_{j=1}^r d_jx_j^2$, with $d_j \in \Fq^*$, $1 \le j \le r$, and $\eta(\prod_{j=1}^rd_j)=\tau$. Thus, there exists a nonsingular linear transformation on $x=(x_1,x_2,\ldots,x_m) \in \Fqm$, converting $Q$ into $\sum_{j=1}^r d_jx_j^2$. Note that $\cL$ is invariant under any nonsingular linear transformation. Thus, this linear transformation converts the set $\{ Q+L+c \mid L \in \cL \}$ into
\begin{equation}
\{\sum_{j=1}^rd_jx_j^2+\sum_{j=1}^m b_jx_j+c \mid b_1,b_2,\ldots,b_m \in \Fq \}. \label{mark-altsetodd}
\end{equation}
Consequently, it suffices to compute $N(f)$, when $f$ ranges over the set (\ref{mark-altsetodd}).

Define a function $f:=f_{b_1,b_2,\ldots,b_m,c}$ from $\Fqm$ to $\Fq$ as
$$
f(x)=\sum_{j=1}^rd_jx_j^2+\sum_{j=1}^m b_jx_j+c, \forall x=(x_1,x_2,\ldots,x_m) \in \Fqm.
$$
When $(b_1,b_2,\ldots,b_m)$ ranges over $\Fqm$, $f$ ranges over the set (\ref{mark-altsetodd}).

If one of $b_{r+1},b_{r+2},\ldots,$ $b_m$ is nonzero, then $N(f)=q^{m-1}$. Note that there are $q^{m}-q^{r}$ choices of $(b_1,b_2,\ldots,b_m)$ such that one of $b_{r+1},b_{r+2},\ldots,b_m$ is nonzero. Then we have
\begin{equation}\label{eqn-oddnonzero}
N(f)=q^{m-1}, \quad \mbox{$q^m-q^r$ times.}
\end{equation}

If $b_{r+1}=b_{r+2}=\cdots=b_m=0$, then $f(x)=\sum_{j=1}^rd_jx_j^2+\sum_{j=1}^r b_jx_j+c$. In the equation $f(x)=0$, we replace $x_j$ with $y_j-\frac{b_j}{2d_j}$, where $1 \le j \le r$. Consequently, we have $f(x)=0$ is equivalent to
\begin{equation}\label{eqn-keyeqnodd}
\sum_{j=1}^rd_jy_j^2=\sum_{j=1}^{r}\frac{b_j^2}{4d_j}-c.
\end{equation}
Hence, in order to compute $N(f)$, it suffices to compute the number of zeroes to the above equation. The general idea is as follows. Note that in the right hand side of (\ref{eqn-keyeqnodd}), $\sum_{j=1}^{r}\frac{b_j^2}{4d_j}$ can be viewed as a quadratic form on $\Fq^r$ with respect to $b_1,b_2,\ldots,b_r$, with rank $r$ and type $\eta(\prod_{j=1}^r \frac{1}{4d_j})=\tau$. When $(b_1,b_2,\ldots,b_r)$ ranges over $\Fq^{r}$, employing Proposition~\ref{prop-numsoluodd}, we can obtain the values and their frequencies in the right hand side of (\ref{eqn-keyeqnodd}). Note that $\sum_{j=1}^rd_jy_j^2$ is a quadratic form of rank $r$ and type $\tau$ over $\Fqm$, employing Proposition~\ref{prop-numsoluodd} again, we can obtain the number of zeroes to (\ref{eqn-keyeqnodd}). Below, Proposition~\ref{prop-numsoluodd} will be frequently used without further reference. We are going to split our discussion into two cases where $r$ is odd or even.

(1) Let $r$ be odd. In this case, there are $q^{r-1}+\tau\eta(-1)^{\frac{r-1}{2}}\eta(h)q^{\frac{r-1}{2}}$ tuples of $(b_1,b_2,\ldots,b_r)$ such that $\sum_{j=1}^r\frac{b_j^2}{4d_j}=h \in \Fq$. By (\ref{eqn-keyeqnodd}), for $h \in \Fq$, we have
\begin{align*}
N(f)=q^{m-1}+\tau\eta(-1)^{\frac{r-1}{2}}&\eta(h-c)q^{m-\frac{r+1}{2}}, \\
  &\mbox{if $r$ is odd, $q^{r-1}+\tau\eta(-1)^{\frac{r-1}{2}}\eta(h)q^{\frac{r-1}{2}}$ times.}
\end{align*}
Therefore, it suffices to calculate the size of intersections between $\{(h-c,h) \mid h \in \Fq\}$ and $U_i$, $1 \le i \le 9$, defined in Lemma~\ref{lem-intersize}. Combining (\ref{eqn-oddnonzero}) and Lemma~\ref{lem-intersize}, a simple computation completes the $r$ odd case.

(2) Let $r$ be even. In this case, there are $q^{r-1}+\tau\eta(-1)^{\frac{r}{2}}\ups(h)q^{\frac{r-2}{2}}$ tuples of $(b_1,b_2,\ldots,b_r)$ such that $\sum_{j=1}^r\frac{b_j^2}{4d_j}=h \in \Fq$. By (\ref{eqn-keyeqnodd}), for $h \in \Fq$, we have
\begin{align*}
N(f)=
  q^{m-1}+\tau\eta(-1)^{\frac{r}{2}}\ups(h-c)&q^{m-\frac{r+2}{2}} \\
  &\mbox{if $r$ is even, $q^{r-1}+\tau\eta(-1)^{\frac{r}{2}}\ups(h)q^{\frac{r-2}{2}}$ times.}
\end{align*}
Therefore, it suffices to calculate the size of the intersections between $\{(h-c,h) \mid h \in \Fq\}$ and the following sets:
$$
\{0\} \times \{0\}, \quad \{0\} \times \Fq^*, \quad \Fq^* \times \{0\}, \quad \Fq^* \times \Fq^*.
$$
A simple computation, together with (\ref{eqn-oddnonzero}), completes the $r$ even case.
\end{proof}

As a direct consequence of Lemma~\ref{lem-quadodd}, we have the following lemma, which is used in the proof of Proposition~\ref{prop-Qweightodd}.

\begin{lemma}\label{lem-quadoddset}
Let $q$ be an odd prime power and $Q$ be a quadratic form of rank $r \ge 1$ and type $\tau$ on $\Fqm$. Let $\cL$ be the set of all homogenous linear functions on $\Fqm$.
\noindent
(1) Let $r$ be odd. When $f$ ranges over $\{Q+L \mid L \in \cL\}$, we have
$$
N(f)=\left\{
\begin{aligned}
  &q^{m-1} & \mbox{$q^m-q^{r}+q^{r-1}$ times,} \\
  &q^{m-1}\pm\tau\eta(-1)^{\frac{r-1}{2}}q^{m-\frac{r+1}{2}} & \mbox{$\frac{(q-1)}{2}(q^{r-1}\pm\tau\eta(-1)^{\frac{r-1}{2}}q^{\frac{r-1}{2}})$ times.}
\end{aligned}\right.
$$
When $f$ ranges over $\{Q+L+c \mid L \in \cL, c \in \Fq^*\}$, we have
$$
N(f)=\left\{
\begin{aligned}
  &q^{m-1} & \mbox{$(q-1)(q^m-q^{r}+q^{r-1})$ times,} \\
  &q^{m-1}\pm\tau\eta(-1)^{\frac{r-1}{2}}q^{m-\frac{r+1}{2}} & \mbox{$\frac{q-1}{2}((q-1)q^{r-1}\mp\tau\eta(-1)^{\frac{r-1}{2}}q^{\frac{r-1}{2}})$ times.}
\end{aligned}\right.
$$
(2) Let $r$ be even. When $f$ ranges over $\{Q+L \mid L \in \cL\}$, we have
$$
N(f)=\left\{
\begin{aligned}
  &q^{m-1} & \mbox{$q^m-q^{r}$ times,} \\
  &q^{m-1}+\tau\eta(-1)^{\frac{r}{2}}q^{m-\frac{r+2}{2}}(q-1) &\mbox{$q^{r-1}+\tau\eta(-1)^{\frac{r}{2}}q^{\frac{r-2}{2}}(q-1)$ times,}\\
  &q^{m-1}-\tau\eta(-1)^{\frac{r}{2}}q^{m-\frac{r+2}{2}} & \mbox{$(q-1)(q^{r-1}-\tau\eta(-1)^{\frac{r}{2}}q^{\frac{r-2}{2}})$ times.}
\end{aligned}\right.
$$
When $f$ ranges over $\{Q+L+c \mid L \in \cL, c \in \Fq^*\}$, we have
$$
N(f)=\left\{
\begin{aligned}
  q^{m-1} \quad\quad\quad\quad\quad &\mbox{$(q-1)(q^m-q^{r})$ times,} \\
  q^{m-1}+\tau\eta(-1)^{\frac{r}{2}}&q^{m-\frac{r+2}{2}}(q-1) \\
  &\mbox{$(q-1)(q^{r-1}-\tau\eta(-1)^{\frac{r}{2}}q^{\frac{r-2}{2}})$ times,}\\
  q^{m-1}-\tau\eta(-1)^{\frac{r}{2}}&q^{m-\frac{r+2}{2}} \\
  &\mbox{$(q-1)((q-1)q^{r-1}+\tau\eta(-1)^{\frac{r}{2}}q^{\frac{r-2}{2}})$ times.}
\end{aligned}\right.
$$
\end{lemma}

Next, we consider the case with $q$ being an even prime power. For $i \in\F_2$, define $T_i=\{ x\in \Fq \mid \Tr^q_2(x)=i \}$. We have the following lemma.

\begin{lemma}\label{lem-quadeven}
Let $q$ be an even prime power and $Q$ be a quadratic form on $\Fqm$. Let $\cL$ be the set of all homogenous linear functions on $\Fqm$ and $c \in \Fq$. Suppose $f$ ranges over $\{Q+L+c \mid L \in \cL\}$. Then the following holds.

\noindent
(1) Let $Q$ have rank $2r$ and type $0$. If $c=0$, then
$$
N(f)=\left\{
\begin{aligned}
  &q^{m-1} & \mbox{$q^m-q^{2r}$ times,} \\
  &q^{m-1}+q^{m-r-1}(q-1) & \mbox{$q^{2r-1}+q^{r-1}(q-1)$ times,}\\
  &q^{m-1}-q^{m-r-1} & \mbox{$(q-1)(q^{2r-1}-q^{r-1})$ times.}
\end{aligned}\right.
$$
If $c \ne 0$, then
$$
N(f)=\left\{
\begin{aligned}
  &q^{m-1} & \mbox{$q^m-q^{2r}$ times,} \\
  &q^{m-1}+q^{m-r-1}(q-1) & \mbox{$q^{2r-1}-q^{r-1}$ times,}\\
  &q^{m-1}-q^{m-r-1} & \mbox{$(q-1)q^{2r-1}+q^{r-1}$ times.}
\end{aligned}\right.
$$
(2) Let $Q$ have rank $2r+1$ and type $1$. If $c=0$, then
$$
N(f)=\left\{
\begin{aligned}
  &q^{m-1} & \mbox{$q^m-q^{2r+1}+q^{2r}$ times,} \\
  &q^{m-1}\pm q^{m-r-1} & \mbox{$(q-1)\frac{q^{2r}\pm q^r}{2}$ times.}
\end{aligned}\right.
$$
If $c \ne 0$, then
$$
N(f)=\left\{
\begin{aligned}
  &q^{m-1} & \mbox{$q^m-q^{2r+1}+q^{2r}$ times,} \\
  &q^{m-1}\pm q^{m-r-1} & \mbox{$\frac{q^{2r+1}-q^{2r}\mp q^r}{2}$ times.}
\end{aligned}\right.
$$
(3) Let $Q$ have rank $2r$ and type $2$. If $c=0$, then
$$
N(f)=\left\{
\begin{aligned}
  &q^{m-1} & \mbox{$q^m-q^{2r}$ times,} \\
  &q^{m-1}-q^{m-r-1}(q-1) & \mbox{$q^{2r-1}-q^{r-1}(q-1)$ times,}\\
  &q^{m-1}+q^{m-r-1} & \mbox{$(q-1)(q^{2r-1}+q^{r-1})$ times.}
\end{aligned}\right.
$$
If $c \ne 0$, then
$$
N(f)=\left\{
\begin{aligned}
  &q^{m-1} & \mbox{$q^m-q^{2r}$ times,} \\
  &q^{m-1}-q^{m-r-1}(q-1) & \mbox{$q^{2r-1}+q^{r-1}$ times,}\\
  &q^{m-1}+q^{m-r-1} & \mbox{$(q-1)q^{2r-1}-q^{r-1}$ times.}
\end{aligned}\right.
$$
\end{lemma}
\begin{proof}
By 2) of Proposition~\ref{prop-equivalence}, $Q$ is equivalent to one of the following canonical quadratic forms
$$
\begin{cases}
\sum_{j=1}^r x_{2j-1}x_{2j} & \mbox{if $Q$ has rank $2r$ and type $0$,} \\
\sum_{j=1}^r x_{2j-1}x_{2j}+x_{2r+1}^2 & \mbox{if $Q$ has rank $2r+1$ and type $1$,} \\
\sum_{j=1}^{r} x_{2j-1}x_{2j}+x_{2r-1}^2+\la x_{2r}^2, \Tr^q_2(\la)=1 & \mbox{if $Q$ has rank $2r$ and type $2$.} \\
\end{cases}
$$
Thus, there exists a nonsingular linear transformation on $x=(x_1,x_2,\ldots,x_m) \in \Fqm$, converting $Q$ into one of the above forms. Note that $\cL$ is invariant under any nonsingular linear transformation. Thus, this linear transformation converts the set $\{ Q+L+c \mid L \in \cL \}$ into one of the following


\begin{equation}\label{mark-altseteven}
\left\{
\begin{aligned}
\{\sum_{j=1}^r x_{2j-1}x_{2j}+\sum_{j=1}^m b_jx_j&+c \mid b_1,b_2,\ldots,b_m \in \Fq \} \\
&\mbox{if $Q$ has rank $2r$ and type $0$,}  \\
\{\sum_{j=1}^r x_{2j-1}x_{2j}+x_{2r+1}^2+&\sum_{j=1}^m b_jx_j+c \mid b_1,b_2,\ldots,b_m \in \Fq \} \\
&\mbox{if $Q$ has rank $2r+1$ and type $1$,} \\
\{\sum_{j=1}^{r} x_{2j-1}x_{2j}+x_{2r-1}^2+&\la x_{2r}^2+\sum_{j=1}^m b_jx_j+c \mid b_1,b_2,\ldots,b_m \in \Fq \} \\
&\mbox{if $Q$ has rank $2r$ and type $2$.}
\end{aligned}
\right.
\end{equation}
Consequently, it suffices to compute $N(f)$, when $f$ ranges over one of the set in (\ref{mark-altseteven}).

For $x=(x_1,x_2,\ldots,x_m) \in \Fqm$, define function $f^{i}:=f^{i}_{b_1,b_2,\ldots,b_m,c}$, $0 \le i \le 2$ as
\begin{equation*}
\begin{cases}
f^{0}(x)=\sum_{j=1}^r x_{2j-1}x_{2j}+\sum_{j=1}^m b_jx_j+c, \\
f^{1}(x)=\sum_{j=1}^r x_{2j-1}x_{2j}+x_{2r+1}^2+\sum_{j=1}^m b_jx_j+c, \\
f^{2}(x)=\sum_{j=1}^r x_{2j-1}x_{2j}+x_{2r-1}^2+\la x_{2r}^2+\sum_{j=1}^m b_jx_j+c.
\end{cases}
\end{equation*}
Below, we split our discussion into three cases. We remark that Proposition~\ref{prop-numsolueven} will be frequently used without further reference.

(1) When $Q$ has rank $2r$ and type $0$, consider the equation $f^0(x)=0$. If one of $b_{2r+1},b_{2r+2},\ldots,b_m$ is nonzero, then $N(f^0)=q^{m-1}$. Note that there are $q^{m}-q^{2r}$ choices of $(b_1,b_2,\ldots,b_m)$ such that one of $b_{2r+1},b_{2r+2},\ldots,b_m$ is nonzero. Then we have
\begin{equation}\label{eqn-evennonzero0}
N(f^0)=q^{m-1}, \quad \mbox{$q^m-q^{2r}$ times.}
\end{equation}

If $b_{2r+1}=b_{2r+2}=\cdots=b_m=0$, we have $f^0(x)=\sum_{j=1}^rx_{2j-1}x_{2j}+\sum_{j=1}^{2r} b_jx_j+c$. In the equation $f^0(x)=0$, we replace $x_{2j-1}$ with $y_{2j-1}+b_{2j}$ and $x_{2j}$ with $y_{2j}+b_{2j-1}$, where $1 \le j \le r$. Consequently, we have $f^0(x)=0$ is equivalent to
\begin{equation}\label{eqn-keyeqneven0}
\sum_{j=1}^ry_{2j-1}y_{2j}=\sum_{j=1}^{r}b_{2j-1}b_{2j}+c.
\end{equation}
There are $q^{2r-1}+\ups(h)q^{r-1}$ tuples of $(b_1,b_2,\ldots,b_r)$ such that $\sum_{j=1}^{r}b_{2j-1}b_{2j}=h$ for some $h \in \Fq$. By (\ref{eqn-keyeqneven0}), for $h \in \Fq$, we have
$$
N(f^0)=q^{m-1}+\ups(h+c)q^{m-r-1}, \quad \mbox{$q^{2r-1}+\ups(h)q^{r-1}$ times}.
$$
Therefore, it suffices to calculate the size of the intersections between $\{(h+c,h) \mid h \in \Fq\}$ and the following sets:
$$
\{0\} \times \{0\}, \quad \{0\} \times \Fq^*, \quad \Fq^* \times \{0\}, \quad \Fq^* \times \Fq^*.
$$
A simple computation, together with (\ref{eqn-evennonzero0}), completes the case $Q$ having rank $2r$ and type $0$.

(2) When $Q$ has rank $2r+1$ and type $1$, consider the equation $f^1(x)=0$. If one of $b_{2r+2},b_{2r+3},\ldots,b_m$ is nonzero, then $N(f^1)=q^{m-1}$. Note that there are $q^{m}-q^{2r+1}$ choices of $(b_1,b_2,\ldots,b_m)$ such that one of $b_{2r+2},b_{2r+3},\ldots,b_m$ is nonzero. Then we have
\begin{equation}\label{eqn-evennonzero1}
N(f^1)=q^{m-1}, \quad \mbox{$q^m-q^{2r+1}$ times.}
\end{equation}
If $b_{2r+2}=b_{2r+3}=\cdots=b_m=0$, we consider two cases where $b_{2r+1}=0$ and $b_{2r+1}\ne0$. When $b_{2r+1}=0$, $f^1(x)=\sum_{j=1}^r x_{2j-1}x_{2j}+x_{2r+1}^2+\sum_{j=1}^{2r}b_jx_j+c$. In the equation $f^1(x)=0$, we replace $x_{2j-1}$ with $y_{2j-1}+b_{2j}$ and $x_{2j}$ with $y_{2j}+b_{2j-1}$, where $1 \le j \le r$, and replace $x_{2r+1}$ with $y_{2r+1}$. Consequently, we have $f^1(x)=0$ is equivalent to
$$
\sum_{j=1}^ry_{2j-1}y_{2j}+y_{2r+1}^2=\sum_{j=1}^rb_{2j-1}b_{2j}+c.
$$
Hence,
\begin{equation}\label{eqn-evennonzero2}
N(f^1)=q^{m-1}, \quad \mbox{$q^{2r}$ times.}
\end{equation}
When $b_{2r+1} \ne 0$, $f^1(x)=\sum_{j=1}^r x_{2j-1}x_{2j}+x_{2r+1}^2+\sum_{j=1}^{2r+1}b_jx_j+c$. In the equation $f^1(x)=0$, we replace $x_{2j-1}$ with $y_{2j-1}+b_{2j}$ and $x_{2j}$ with $y_{2j}+b_{2j-1}$ where $1 \le j \le r$, and replace $x_{2r+1}$ with $y_{2r+1}$. Consequently, we have $f^1(x)=0$ is equivalent to
$$
\sum_{j=1}^ry_{2j-1}y_{2j}=y_{2r+1}^2+b_{2r+1}y_{2r+1}+c+\sum_{j=1}^rb_{2j-1}b_{2j}.
$$
Furthermore, replacing $y_j$ with $b_{2r+1}z_j$ for $1 \le j \le 2r+1$, $b_j$ with $b_{2r+1}b_j^{\pr}$ for $1 \le j \le 2r$ and $c$ with $b_{2r+1}^2c^{\pr}$ in the above equation, we have
\begin{equation}\label{eqn-keyeqneven1}
\sum_{j=1}^rz_{2j-1}z_{2j}=z_{2r+1}^2+z_{2r+1}+c^{\pr}+\sum_{j=1}^rb_{2j-1}^{\pr}b_{2j}^{\pr}.
\end{equation}
If $\Tr^q_2(c^{\pr})=0$, then the multiset $\{z_{2r+1}^2+z_{2r+1}+c^{\pr} \mid z_{2r+1} \in \Fq\}$ consists of each element of $T_0$ exactly twice. Hence, the multiset $\{z_{2r+1}^2+z_{2r+1}+c^{\pr}+\sum_{j=1}^rb_{2j-1}^{\pr}b_{2j}^{\pr} \mid z_{2r+1} \in \Fq\}$ contains two zero elements and $q-2$ nonzero elements if $\sum_{j=1}^{r}b_{2j-1}^{\pr}b_{2j}^{\pr} \in T_0$, and contains $q$ nonzero elements if $\sum_{j=1}^{r}b_{2j-1}^{\pr}b_{2j}^{\pr} \in T_1$.  Once $z_{2r+1}$ is fixed, we can regard the left hand size of (\ref{eqn-keyeqneven1}) as a quadratic form with rank $2r$ and type $0$ on an $(m-1)$-dimensional vector space over $\Fq$. Note that $f^1(x)=0$ is equivalent to (\ref{eqn-keyeqneven1}). We have $N(f^1)=2(q^{m-2}+q^{m-r-2}(q-1))+(q-2)(q^{m-2}-q^{m-r-2})=q^{m-1}+q^{m-r-1}$ if $\sum_{j=1}^{r}b_{2j-1}^{\pr}b_{2j}^{\pr} \in T_0$ and $N(f^1)=q(q^{m-2}-q^{m-r-2})=q^{m-1}-q^{m-r-1}$ if $\sum_{j=1}^{r}b_{2j-1}^{\pr}b_{2j}^{\pr} \in T_1$. There are $q^{2r-1}+q^{r-1}(q-1)+(\frac{q}{2}-1)(q^{2r-1}-q^{r-1})=\frac{q^{2r}+q^r}{2}$ choices of $(b_1^{\pr},b_2^{\pr},\ldots,b_{2r}^{\pr})$ such that $\sum_{j=1}^{r}b_{2j-1}^{\pr}b_{2j}^{\pr} \in T_0$ and $\frac{q}{2}(q^{2r-1}-q^{r-1})=\frac{q^{2r}-q^r}{2}$ choices of $(b_1^{\pr},b_2^{\pr},\ldots,b_{2r}^{\pr})$ such that $\sum_{j=1}^{r}b_{2j-1}^{\pr}b_{2j}^{\pr} \in T_1$.
Hence
\begin{equation}\label{eqn-trace0}
N(f^1)=\left\{
\begin{aligned}
  q^{m-1}+&q^{m-r-1} \\
  &\mbox{if $b_{2r+1} \ne 0$, $\Tr^q_2(c^{\pr})=0$, $\sum_{j=1}^{r}b_{2j-1}^{\pr}b_{2j}^{\pr} \in T_0$, $\frac{q^{2r}+q^r}{2}$ times,}\\
  q^{m-1}-&q^{m-r-1} \\
  &\mbox{if $b_{2r+1} \ne 0$, $\Tr^q_2(c^{\pr})=0$, $\sum_{j=1}^{r}b_{2j-1}^{\pr}b_{2j}^{\pr} \in T_1$, $\frac{q^{2r}-q^r}{2}$ times.}
\end{aligned}
\right.
\end{equation}

If $\Tr^q_2(c^{\pr})=1$, the multiset $\{z_{2r+1}^2+z_{2r+1}+c^{\pr} \mid z_{2r+1} \in \Fq \}$ consists of each element of $T_1$ exactly twice. Using a similar approach, we have
\begin{equation}\label{eqn-trace1}
N(f^1)=\left\{
\begin{aligned}
  q^{m-1}+&q^{m-r-1}  \\
     &\mbox{if $b_{2r+1} \ne 0$, $\Tr^q_2(c^{\pr})=1$, $\sum_{j=1}^{r}b_{2j-1}^{\pr}b_{2j}^{\pr} \in T_1$, $\frac{q^{2r}-q^{r}}{2}$ times,} \\
  q^{m-1}-&q^{m-r-1}  \\
     &\mbox{if $b_{2r+1} \ne 0$, $\Tr^q_2(c^{\pr})=1$, $\sum_{j=1}^{r}b_{2j-1}^{\pr}b_{2j}^{\pr} \in T_0$, $\frac{q^{2r}+q^{r}}{2}$ times.}
\end{aligned}
\right.
\end{equation}

Recall that $c^{\pr}=\frac{c}{b_{2r+1}^2}$, if $c=0$, when $b_{2r+1}$ ranges over $\Fq^*$, $c^{\pr}$ take the $0$ value $q-1$ times. Therefore, $\Tr^q_2(c^{\pr})$ takes the $0$ value $q-1$ times. Combining (\ref{eqn-evennonzero1}), (\ref{eqn-evennonzero2}) and (\ref{eqn-trace0}), we complete the $c=0$ case. Similarly, if $c \ne 0$, when $b_{2r+1}$ ranges over $\Fq^*$, $c^{\pr}$ ranges over $\Fq^*$ as well. Therefore, $\Tr^q_2(c^{\pr})$ takes the $0$ value $\frac{q}{2}-1$ times and the $1$ value $\frac{q}{2}$ times. Combining (\ref{eqn-evennonzero1}), (\ref{eqn-evennonzero2}), (\ref{eqn-trace0}) and (\ref{eqn-trace1}), we complete the $c \ne 0$ case.

(3) When $Q$ has rank $2r$ and type $2$, consider the equation $f^2(x)=0$. If one of $b_{2r+1},b_{2r+2},\ldots,b_m$ is nonzero, then $N(f^2)=q^{m-1}$. Note that there are $q^{m}-q^{2r}$ choices of $(b_1,b_2,\ldots,b_m)$ such that one of $b_{2r+1},b_{2r+2},\ldots,b_m$ is nonzero. Then we have
\begin{equation}\label{eqn-evennonzero3}
N(f^2)=q^{m-1}, \quad \mbox{$q^m-q^{2r}$ times.}
\end{equation}
If $b_{2r+1}=b_{2r+2}=\cdots=b_m=0$, we have $f^2(x)=\sum_{j=1}^rx_{2j-1}x_{2j}+x_{2r-1}^2+\la x_{2r}^2+\sum_{j=1}^{2r} b_jx_j+c$. In the equation $f^2(x)=0$, we replace $x_{2j-1}$ with $y_{2j-1}+b_{2j}$ and $x_{2j}$ with $y_{2j}+b_{2j-1}$, where $1 \le j \le r$. Consequently, we have $f^2(x)=0$ is equivalent to
\begin{equation*}
\sum_{j=1}^ry_{2j-1}y_{2j}+y_{2r-1}^2+\la y_{2r}^2=\sum_{j=1}^{r}b_{2j-1}b_{2j}+b_{2r}^2+\la b_{2r-1}^2+c.
\end{equation*}
There are $q^{2r-1}-\ups(h)q^{r-1}$ tuples of $(b_1,b_2,\ldots,b_{2r})$ such that $\sum_{j=1}^{r}b_{2j-1}b_{2j}+b_{2r}^2+\la b_{2r-1}^2=h$ for some $h \in \Fq$. When $c=0$, the multiset $\{\sum_{j=1}^{r}b_{2j-1}b_{2j}+b_{2r}^2+\la b_{2r-1}^2 \mid b_1,b_2
,\ldots,b_{2r} \in \Fq\}$ consists of $q^{2r-1}-q^{r-1}(q-1)$ zero elements and $(q-1)(q^{2r-1}+q^{r-1})$ nonzero elements. Hence,
$$
N(f^2)=\begin{cases}
q^{m-1}-q^{m-r-1}(q-1), & \mbox{if $c=0$, $q^{2r-1}-q^{r-1}(q-1)$ times,} \\
q^{m-1}+q^{m-r-1}, & \mbox{if $c=0$, $(q-1)(q^{2r-1}+q^{r-1})$ times.}
\end{cases}
$$
Together with (\ref{eqn-evennonzero3}), we complete the $c=0$ case.

When $c\ne0$, the multiset $\{\sum_{j=1}^{r}b_{2j-1}b_{2j}+b_{2r}^2+\la b_{2r-1}^2+c \mid b_1,b_2,\ldots,b_{2r} \in \Fq\}$ consists of $c$ for $q^{2r-1}-q^{r-1}(q-1)$ times and each element in $\Fq\sm \{c\}$ for $q^{2r-1}+q^{r-1}$ times. Thus, it contains $q^{2r-1}+q^{r-1}$ zero elements and $q^{2r-1}-q^{r-1}(q-1)+(q-2)(q^{2r-1}+q^{r-1})=(q-1)q^{2r-1}-q^{r-1}$ nonzero elements. Hence,
$$
N(f^2)=\left\{
\begin{aligned}
&q^{m-1}-q^{m-r-1}(q-1) \quad &\mbox{if $c\ne0$, $(q-1)(q^{2r-1}+q^{r-1})$ times,} \\
&q^{m-1}+q^{m-r-1} \quad &\mbox{if $c\ne0$, $(q-1)((q-1)q^{2r-1}-q^{r-1})$ times.}
\end{aligned}
\right.
$$
Together with (\ref{eqn-evennonzero3}), we complete the $c \ne 0$ case.
\end{proof}

As a direct consequence of Lemma~\ref{lem-quadeven}, we have the following lemma, which is used in the proof of Proposition~\ref{prop-Qweighteven}.

\begin{lemma}\label{lem-quadevenset}
Let $q$ be an even prime power and $Q$ be a quadratic form on $\Fqm$. Let $\cL$ be the set of all homogenous linear functions on $\Fqm$.

\noindent
(1) Let $Q$ have rank $2r$ and type $0$. When $f$ ranges over $\{Q+L \mid L \in \cL\}$, we have
$$
N(f)=\left\{
\begin{aligned}
  &q^{m-1} & \mbox{$q^m-q^{2r}$ times,} \\
  &q^{m-1}+q^{m-r-1}(q-1) & \mbox{$q^{2r-1}+q^{r-1}(q-1)$ times,}\\
  &q^{m-1}-q^{m-r-1} & \mbox{$(q-1)(q^{2r-1}-q^{r-1})$ times.}
\end{aligned}\right.
$$
When $f$ ranges over $\{Q+L+c \mid L \in \cL, c \in \Fq^*\}$, we have
$$
N(f)=\left\{
\begin{aligned}
  &q^{m-1} & \mbox{$(q-1)(q^m-q^{2r})$ times,} \\
  &q^{m-1}+q^{m-r-1}(q-1) & \mbox{$(q-1)(q^{2r-1}-q^{r-1})$ times,}\\
  &q^{m-1}-q^{m-r-1} & \mbox{$(q-1)((q-1)q^{2r-1}+q^{r-1})$ times.}
\end{aligned}\right.
$$
(2) Let $Q$ have rank $2r+1$ and type $1$. When $f$ ranges over $\{Q+L \mid L \in \cL\}$, we have
$$
N(f)=\left\{
\begin{aligned}
  &q^{m-1} & \mbox{$q^m-q^{2r+1}+q^{2r}$ times,} \\
  &q^{m-1}\pm q^{m-r-1} & \mbox{$(q-1)\frac{q^{2r}\pm q^r}{2}$ times.}
\end{aligned}\right.
$$
When $f$ ranges over $\{Q+L+c \mid L \in \cL, c \in \Fq^*\}$, we have
$$
N(f)=\left\{
\begin{aligned}
  &q^{m-1} & \mbox{$(q-1)(q^m-q^{2r+1}+q^{2r})$ times,} \\
  &q^{m-1}\pm q^{m-r-1} & \mbox{$(q-1)\frac{q^{2r+1}-q^{2r}\mp q^r}{2}$ times.}
\end{aligned}\right.
$$
(3) Let $Q$ have rank $2r$ and type $2$. When $f$ ranges over $\{Q+L \mid L \in \cL\}$, we have
$$
N(f)=\left\{
\begin{aligned}
  &q^{m-1} & \mbox{$q^m-q^{2r}$ times,} \\
  &q^{m-1}-q^{m-r-1}(q-1) & \mbox{$q^{2r-1}-q^{r-1}(q-1)$ times,}\\
  &q^{m-1}+q^{m-r-1} & \mbox{$(q-1)(q^{2r-1}+q^{r-1})$ times.}
\end{aligned}\right.
$$
When $f$ ranges over $\{Q+L+c \mid L \in \cL, c \in \Fq^*\}$, we have
$$
N(f)=\left\{
\begin{aligned}
  &q^{m-1} & \mbox{$(q-1)(q^m-q^{2r})$ times,} \\
  &q^{m-1}-q^{m-r-1}(q-1) & \mbox{$(q-1)(q^{2r-1}+q^{r-1})$ times,}\\
  &q^{m-1}+q^{m-r-1} & \mbox{$(q-1)((q-1)q^{2r-1}-q^{r-1})$ times.}
\end{aligned}\right.
$$
\end{lemma}

In fact, we already determined the number of solutions to the equation $Q(x)+L(x)+c=0$ in the proofs of Lemmas~\ref{lem-quadodd} and \ref{lem-quadeven}, where $Q$ is a canonical quadratic form. We summarize this result in the following proposition. Although this proposition is not used in the paper, we believe it is worthwhile to be documented.

\begin{proposition}
Suppose $Q$ is a quadratic form on $\Fqm$, $L(x)=\sum_{j=1}^mb_jx_j$ and $c \in \Fq$. Let $f=Q+L+c$. Then we have the following.

\noindent
(1) Suppose $q$ is an odd prime power and $Q$ is a canonical quadratic form with rank $r$ and type $\tau$, i.e., $Q(x)=\sum_{j=1}^rd_jx_j^2$ with $d_j \in \Fq^*$, $1 \le j \le r$, and $\tau=\eta(\prod_{j=1}^r d_j)$. Then, when $r$ is odd,
$$
N(f)=\left\{
\begin{aligned}
  &q^{m-1} \quad \mbox{if one of $b_i$, $r+1 \le i \le m$, nonzero}, \\
  &q^{m-1}+\tau\eta(-1)^{\frac{r-1}{2}}\eta(\sum_{j=1}^r\frac{b_j^2}{4d_j}-c)q^{m-\frac{r+1}{2}} \quad  \mbox{if $b_i=0$, $r+1 \le i \le m$.}
\end{aligned}\right.
$$
When $r$ is even,
$$
N(f)=\left\{
\begin{aligned}
  &q^{m-1} \quad \mbox{if one of $b_i$, $r+1 \le i \le m$, nonzero}, \\
  &q^{m-1}+\tau\eta(-1)^{\frac{r}{2}}\ups(\sum_{j=1}^r\frac{b_j^2}{4d_j}-c)q^{m-\frac{r+2}{2}} \quad  \mbox{if $b_i=0$, $r+1 \le i \le m$.}
\end{aligned}\right.
$$
\end{proposition}

\noindent
(2) Suppose $q$ is an even prime power and $Q$ is a canonical quadratic form. Then, when $Q$ has rank $2r$ and type $0$, i.e., $Q(x)=\sum_{j=1}^r x_{2j-1}x_{2j}$,
$$
N(f)=\left\{
\begin{aligned}
  &q^{m-1} \quad \mbox{if one of $b_i$, $2r+1 \le i \le m$, nonzero}, \\
  &q^{m-1}+\ups(\sum_{j=1}^rb_{2j-1}b_{2j}+c)q^{m-r-1} \quad  \mbox{if $b_i=0$, $2r+1 \le i \le m$.}
\end{aligned}\right.
$$
When $Q$ has rank $2r+1$ and type $1$, i.e., $Q(x)=\sum_{j=1}^r x_{2j-1}x_{2j}+x_{2r+1}^2$,
$$
N(f)=\left\{
\begin{aligned}
  &q^{m-1} \quad  \mbox{if one of $b_i$, $2r+2 \le i \le m$, nonzero}, \\
  &  \quad \quad \quad \quad  \mbox{or $b_i=0$, $2r+1 \le i \le m$}, \\
  &q^{m-1}+q^{m-r-1} \quad  \mbox{if $b_{2r+1}\ne0$, $b_i=0$, $2r+2 \le i \le m$}, \\
  &   \quad \quad \quad \quad \quad \quad \quad \quad \quad \mbox{and $\Tr^q_2(\frac{c+\sum_{j=1}^rb_{2j-1}b_{2j}}{b_{2r+1}^2})=0$}, \\
  &q^{m-1}-q^{m-r-1} \quad  \mbox{if $b_{2r+1}\ne0$, $b_i=0$, $2r+2 \le i \le m$}, \\
  &   \quad \quad \quad \quad \quad \quad \quad \quad \quad \mbox{and $\Tr^q_2(\frac{c+\sum_{j=1}^rb_{2j-1}b_{2j}}{b_{2r+1}^2})=1$}.
\end{aligned}\right.
$$
When $Q$ has rank $2r$ and type $2$, i.e., $Q(x)=\sum_{j=1}^{r} x_{2j-1}x_{2j}+x_{2r-1}^2+\la x_{2r}^2$, for some $\Tr^q_2(\la)=1$,
$$
N(f)=\left\{
\begin{aligned}
  &q^{m-1} \quad \mbox{if one of $b_i$, $2r+1 \le i \le m$, nonzero}, \\
  &q^{m-1}-\ups(\sum_{j=1}^rb_{2j-1}b_{2j}+b_{2r}^2+\la b_{2r-1}^2+c)q^{m-r-1} \\
  & \quad \quad \quad \quad \quad \quad \quad \quad \quad \quad \quad \quad \mbox{if $b_i=0$, $2r+1 \le i \le m$}.
\end{aligned}\right.
$$

\section*{Acknowledgments}

The author wishes to thank the Associated Editor, Dr. Ron Roth and the anonymous reviewers for their valuable comments, which greatly improve the presentation of this paper. He is indebted to Jonathan Jedwab, for many helpful suggestions on the first version of the manuscript. The author would like to express his deepest gratitude to Prof. Gennian Ge, Capital Normal University, and Dr. Tao Feng, Zhejiang University, for their guidance and encouragement.


\begin{thebibliography}{10}

\bibitem{AKS}
{\sc S.~A. Aly, A.~Klappenecker, and P.~K. Sarvepalli}, {\em On quantum and
  classical {BCH} codes}, IEEE Trans. Inform. Theory, 53 (2007),
  pp.~1183--1188.

\bibitem{AK}
{\sc E.~F. Assmus and J.~D. Key}, {\em Designs and their codes}, vol.~103 of
  Cambridge Tracts in Mathematics, Cambridge University Press, Cambridge, 1992.

\bibitem{ACS}
{\sc D.~Augot, P.~Charpin, and N.~Sendrier}, {\em Studying the locator
  polynomials of minimum weight codewords of {BCH} codes}, IEEE Trans. Inform.
  Theory, 38 (1992), pp.~960--973.

\bibitem{AS}
{\sc D.~Augot and N.~Sendrier}, {\em Idempotents and the {BCH} bound}, IEEE
  Trans. Inform. Theory, 40 (1994), pp.~204--207.

\bibitem{BI}
{\sc E.~Bannai and T.~Ito}, {\em Algebraic combinatorics {I}: association
  schemes}, The Benjamin/Cummings Publishing Co., Menlo Park, CA, 1984.

\bibitem{Ber67}
{\sc E.~R. Berlekamp}, {\em The enumeration of information symbols in {BCH}
  codes}, Bell System Tech. J., 46 (1967), pp.~1861--1880.

\bibitem{Ber70}
{\sc E.~R. Berlekamp}, {\em The weight enumerators for certain subcodes of the
  second order binary {R}eed-{M}uller codes}, Information and Control, 17
  (1970), pp.~485--500.

\bibitem{Ber15}
{\sc E.~R. Berlekamp}, {\em Algebraic coding theory}, World Scientific
  Publishing Co. Pte. Ltd., Hackensack, NJ, revised~ed., 2015.

\bibitem{CC}
{\sc A.~Canteaut and F.~Chabaud}, {\em A new algorithm for finding
  minimum-weight words in a linear code: application to {M}c{E}liece's
  cryptosystem and to narrow-sense {BCH} codes of length 511}, IEEE Trans.
  Inform. Theory, 44 (1998), pp.~367--378.

\bibitem{Char90}
{\sc P.~Charpin}, {\em On a class of primitive {BCH}-codes}, IEEE Trans.
  Inform. Theory, 36 (1990), pp.~222--228.

\bibitem{Char98}
{\sc P.~Charpin}, {\em Open problems on cyclic codes}, in Handbook of coding
  theory, {V}ol. {I}, North-Holland, Amsterdam, 1998, pp.~963--1063.

\bibitem{Co}
{\sc G.~Cohen}, {\em On the minimum distance of some {BCH} codes}, IEEE Trans.
  Inform. Theory, 26 (1980), p.~363.

\bibitem{Del73}
{\sc P.~Delsarte}, {\em An algebraic approach to the association schemes of
  coding theory}, Philips Res. Rep. Suppl.,  (1973).

\bibitem{Del75}
{\sc P.~Delsarte}, {\em On subfield subcodes of modified {R}eed-{S}olomon
  codes}, IEEE Trans. Inform. Theory, 21 (1975), pp.~575--576.

\bibitem{Del78}
{\sc P.~Delsarte}, {\em Bilinear forms over a finite field, with applications
  to coding theory}, J. Combin. Theory Ser. A, 25 (1978), pp.~226--241.

\bibitem{DG}
{\sc P.~Delsarte and J.~M. Goethals}, {\em Alternating bilinear forms over
  {$GF(q)$}}, J. Combin. Theory Ser. A, 19 (1975), pp.~26--50.

\bibitem{Ding15}
{\sc C.~Ding}, {\em Codes from difference sets}, World Scientific Publishing
  Co. Pte. Ltd., Hackensack, NJ, 2015.

\bibitem{Ding2}
{\sc C.~Ding}, {\em Parameters of several classes of {BCH} codes}, IEEE Trans.
  Inform. Theory, 61 (2015), pp.~5322--5330.

\bibitem{Ding1}
{\sc C.~Ding, X.~Du, and Z.~Zhou}, {\em The {B}ose and minimum distance of a
  class of {BCH} codes}, IEEE Trans. Inform. Theory, 61 (2015), pp.~2351--2356.

\bibitem{Ding3}
{\sc C.~Ding, C.~Fan, and Z.~Zhou}, {\em The dimension and minimum distance of
  two classes of primitive {BCH} codes}, Finite Fields Appl., 45 (2017),
  pp.~237--263.

\bibitem{FWMM}
{\sc R.~Feng, Y.~Wang, C.~Ma, and J.~Ma}, {\em Eigenvalues of association
  schemes of quadratic forms}, Discrete Math., 308 (2008), pp.~3023--3047.

\bibitem{Gra}
{\sc M.~Grassl}, {\em {Bounds on the minimum distance of linear codes and
  quantum codes}}.
\newblock Online available at \url{http://www.codetables.de}, 2007.

\bibitem{HW}
{\sc Y.~J. Huo and Z.~X. Wan}, {\em Nonsymmetric association schemes of
  symmetric matrices}, Acta Math. Appl. Sinica (English Ser.), 9 (1993),
  pp.~236--255.

\bibitem{Jung}
{\sc D.~Jungnickel}, {\em Difference sets}, in Contemporary design theory,
  Wiley-Intersci. Ser. Discrete Math. Optim., Wiley, New York, 1992,
  pp.~241--324.

\bibitem{K1}
{\sc T.~Kasami}, {\em Weight distributions of {B}ose-{C}haudhuri-{H}ocquenghem
  codes}, in Combinatorial {M}athematics and its {A}pplications ({P}roc.
  {C}onf., {U}niv. {N}orth {C}arolina, {C}hapel {H}ill, {N}.{C}., 1967),
  pp.~335--357.

\bibitem{K2}
{\sc T.~Kasami}, {\em The weight enumerators for several classes of subcodes of
  the {$2$}nd order binary {R}eed-{M}uller codes}, Information and Control, 18
  (1971), pp.~369--394.

\bibitem{KL}
{\sc T.~Kasami and S.~Lin}, {\em Some results on the minimum weight of
  primitive {BCH} codes (corresp.)}, IEEE Trans. Inform. Theory, 18 (1972),
  pp.~824--825.

\bibitem{KLP1}
{\sc T.~Kasami, S.~Lin, and W.~W. Peterson}, {\em Linear codes which are
  invariant under the affine group and some results on minimum weights in {BCH}
  codes}, Electron. Commun. Japan, 50 (1967), pp.~100--106.

\bibitem{KT}
{\sc T.~Kasami and N.~Tokura}, {\em Some remarks on {BCH} bounds and minimum
  weights of binary primitive {BCH} codes}, IEEE Trans. Inform. Theory, 15
  (1969), pp.~408--413.

\bibitem{LDXG}
{\sc S.~Li, C.~Ding, M.~Xiong, and G.~Ge}, {\em Narrow-sense {BCH} codes over
  {GF}($q$) with length $n=\frac{q^m-1}{q-1}$}, IEEE Trans. Inform. Theory, 63
  (2017), pp.~7219--7236.

\bibitem{LN}
{\sc R.~Lidl and H.~Niederreiter}, {\em Finite fields}, vol.~20 of Encyclopedia
  of Mathematics and its Applications, Addison-Wesley Publishing Company,
  Advanced Book Program, Reading, MA, 1983.

\bibitem{LF}
{\sc J.~Luo and K.~Feng}, {\em On the weight distributions of two classes of
  cyclic codes}, IEEE Trans. Inform. Theory, 54 (2008), pp.~5332--5344.

\bibitem{Ma}
{\sc S.~L. Ma}, {\em A survey of partial difference sets}, Des. Codes
  Cryptogr., 4 (1994), pp.~221--261.

\bibitem{MS}
{\sc F.~J. MacWilliams and N.~J.~A. Sloane}, {\em The theory of
  error-correcting codes}, North-Holland Publishing Co., Amsterdam-New
  York-Oxford, 1977.
\newblock North-Holland Mathematical Library, Vol. 16.

\bibitem{Mann}
{\sc H.~B. Mann}, {\em On the number of information symbols in
  {B}ose-{C}haudhuri codes}, Information and Control, 5 (1962), pp.~153--162.

\bibitem{McE}
{\sc R.~J. McEliece}, {\em Quadratic forms over finite fields and second-order
  {R}eed-{M}uller codes}, JPL Space Programs Summary, 3 (1969), pp.~37--58.

\bibitem{Pe}
{\sc W.~W. Peterson}, {\em Some new results on finite fields and their
  application to the theory of {BCH} codes}, in Combinatorial {M}athematics and
  its {A}pplications ({P}roc. {C}onf., {U}niv. {N}orth {C}arolina, {C}hapel
  {H}ill, {N}.{C}., 1967), pp.~329--334.

\bibitem{Sch10}
{\sc K.-U. Schmidt}, {\em Symmetric bilinear forms over finite fields of even
  characteristic}, J. Combin. Theory Ser. A, 117 (2010), pp.~1011--1026.

\bibitem{Sch15}
{\sc K.-U. Schmidt}, {\em Symmetric bilinear forms over finite fields with
  applications to coding theory}, J. Algebraic Combin., 42 (2015),
  pp.~635--670.

\bibitem{WWMM}
{\sc Y.~Wang, C.~Wang, C.~Ma, and J.~Ma}, {\em Association schemes of quadratic
  forms and symmetric bilinear forms}, J. Algebraic Combin., 17 (2003),
  pp.~149--161.

\bibitem{YF}
{\sc D.-W. Yue and G.-Z. Feng}, {\em Minimum cyclotomic coset representatives
  and their applications to {BCH} codes and {G}oppa codes}, IEEE Trans. Inform.
  Theory, 46 (2000), pp.~2625--2628.

\end{thebibliography}
\end{document}